\documentclass{lmcs}

\usepackage{soul}
\usepackage{url}
\usepackage[hidelinks]{hyperref}
\usepackage[utf8]{inputenc}
\usepackage[small]{caption}
\usepackage{graphicx}
\usepackage{amsmath}
\usepackage{amsthm}
\usepackage{booktabs}
\usepackage{algorithm}
\usepackage{algorithmic}
\newtheorem{example}{Example}
\newtheorem{theorem}{Theorem}

\usepackage{amssymb}
\usepackage{amsthm}
\usepackage{ marvosym }

\newtheorem{observation}[theorem]{Observation}

\theoremstyle{definition}

\newtheorem*{problem*}{Decision Problem}
\usepackage{ marvosym }
\usepackage{amsmath}
\usepackage{mathtools}
\usepackage{tabularx}
\usepackage{tikz-cd}
\usepackage{forest}
\usepackage{hyperref}
\hypersetup{
    colorlinks=true,
    linkcolor=blue,
    filecolor=magenta,      
    urlcolor=cyan
}

\newtheorem{manualtheoreminner}{Theorem}
\newenvironment{manualtheorem}[1]{%
  \IfBlankTF{#1}
    {\renewcommand{\themanualtheoreminner}{\unskip}}
    {\renewcommand\themanualtheoreminner{#1}}%
  \manualtheoreminner
}{\endmanualtheoreminner}
\newtheorem{manuallemmainner}{Lemma}
\newenvironment{manuallemma}[1]{%
  \IfBlankTF{#1}
    {\renewcommand{\themanuallemmainner}{\unskip}}
    {\renewcommand\themanuallemmainner{#1}}%
  \manuallemmainner
}{\endmanuallemmainner}

\usepackage{tikz}
\usetikzlibrary{automata, positioning, arrows, shapes}

\usepackage{mathtools}          

\newcommand{\s}{{\tt s}}
\newcommand{\G} {{\tt G}}
 \newcommand{\Salt} {{\tt S}}
\newcommand{\Oalt} {{\tt O}}
\newcommand{\Win}{\mathsf{Win}}

\begin{document}

\title[Modeling Concurrent Multi-Agent Systems]{Modeling Concurrent Multi-Agent Systems}

\author[S. Rajasekaran]{Senthil Rajasekaran\lmcsorcid{0000-0002-6675-8063}}
\author[M. Y. Vardi]{Moshe Y. Vardi\lmcsorcid{0000-0002-0661-5773}}

\address{Université Libre de Bruxelles \thanks{A portion of this work was done while the author was a PhD student at Rice University}}	
\email{senthilrajasekaran@ulb.be}

\address{Rice University Department of Computer Science}	
\email{vardi@rice.edu} 
\pagenumbering{arabic} 
\setcounter{page}{1}

\begin{abstract}
Recent work in the field of multi-agent systems has sought to use techniques and concepts from the field of formal methods to provide rigorous theoretical analysis and guarantees on complex systems where multiple agents strategically interact, leading to the creation of the field of equilibrium analysis, which studies equilibria concepts from the field of game theory through a complexity-theoretic lens. Multi-agent systems, however, are complex mathematical objects, and, therefore, defining them in a precise mathematical manner is non-trivial. As a result, researchers often considered more restrictive models that are easier to model but lack expressive power or simply omitted critical complexity-theoretic results in their analysis. This paper addresses this problem by carefully analyzing and contrasting complexity-theoretic results in the explicit model, a mathematically precise formulation of the models commonly used in the literature, and the circuit-based model, a novel model that addresses the problems found in the literature. The utility of the circuit-based model is demonstrated through a comprehensive analysis that considers upper and lower bounds for the realizability and verification problems, the two most important decision problems in equilibrium analysis, for both models. By conducting this analysis, we see that problematic issues that are endemic to the explicit model and the equilibrium analysis literature as a whole are adequately handled by the circuit-based model.  
\end{abstract}

\maketitle

\section{Introduction}\label{intro}

Systems in which multiple independent agents interact with one another have become ubiquitous in the modern landscape of computer science~\cite{SLmultiagentbook,gullí2025agentic}. This has motivated research into formal models for multi-agent systems, which are mathematical constructions that model systems where independent agents strategically interact with one another. When multiple agents interact with one another, the total number of potential outcomes increases exponentially, potentially leading to unsafe behaviors. Therefore, researchers from the field of \emph{formal verification} have recently begun to analyze multi-agent systems as part of a broad research program to provide rigorous theoretical guarantees to complex multi-agent systems~\cite{wooldridge2009introduction,aibook}. 

The traditional setting in the formal-verification literature consists of a system that aims to enforce a property and an environment assumed to be entirely antagonistic to the system's goals \cite{pnueli1992temporal}. Here, the assumption of antagonism enables a clear mathematical formalism that completely defines the behaviors of both the system and the environment. In multi-agent systems, however, a strictly antagonistic relationship between the agents is no longer a reasonable assumption, as each agent has its own goal and incentives. Therefore, within the execution of a single multi-agent system, the relationship between agents may switch back and forth between being characterized by cooperation, competition, mutual non-interest, and many other modes. To provide a precise mathematical characterization of this complex behavior, researchers in formal methods have turned to \emph{equilibria}, which are solution concepts for multi-player games from the field of \emph{game theory}~\cite{Osborne1994}.  This has led to the development of the subfield of \emph{equilibrium analysis}, which is also called \emph{rational verification} in some contexts~\cite{ratver}, within the formal methods literature.

In equilibrium analysis, the complexity-theoretic properties of decision problems related to equilibria are analyzed in multi-agent systems. The broad motivation is that since equilibria allow for a precise mathematical characterization of the complex interactions between agents in multi-agent systems, studying equilibria in multi-agent systems allows for a precise mathematical analysis of how multi-agent systems scale. Understanding how these systems scale is crucial to identifying the bottlenecks that arise within them and the challenges that arise in providing formal guarantees. This idea has led to many recent research works, see~\cite{GPW17,parkerverification,raskinsubgame2,PRISM,EVE,rajasekaranthesis} for a few examples.

In a multi-agent system, each agent is associated with a ``local state" that details the information that is relevant to the agent. It is not enough, however, to consider the ``total state" of the multi-agent system to be a product of the local states, as it is often necessary to consider an ``environment" which contains global information~\cite{reasoningaboutknowledge}, such as the state of resources that are not controlled by a single agent. The addition of global states also adds interactivity to the multi-agent system, as agents can influence the decisions of other agents through changes in the global state. This type of interactivity is usually mathematically formalized through a \emph{transition table}, which details how the global state updates based on the collective choices of the agents~\cite{reasoningaboutknowledge}

A common assumption is that agents interact with the system and its global states through choices of ``actions". We further assume that each agent in the system has at least two actions to choose from, which allows each agent to make a meaningful choice. This implies that if there are $k$ agents in the system, the transition table must account for the $2^k$ different collective action choices available to the agents collectively. Therefore, in a multi-agent system, the transition can be viewed as a construction that is exponential in the number of agents. This has profound impacts, detailed below, on the complexity-theoretic analysis of multi-agent systems, as the size of the representation of the input is a central issue in the field of complexity theory~\cite{papacomplexity}.

From an applied perspective, this decouples the practical performance of tools from the theoretical results obtained in equilibrium analysis. This can be seen by considering state-of-the-art tools such as~\cite{PRISM,MCMAStool}, which consider transition table models that are worst-case exponential. This worst-case scenario, however, is never realized in the practical applications of these tools to existing systems of interest, as most practical systems have simpler transition rules that do not require considering every possible collective action choice independently (see, for example, the discussion in~\cite{factoredMDP}). This suggests a gap between theoretical analysis and practical performance. This gap limits the applicability of the complexity-theoretic results of equilibrium analysis to the performance of multi-agent verification tools.  

From a theoretical perspective, which is our focus in this paper, the most significant issue arising from the size of the transition table is the lack of lower bounds in the field of equilibrium analysis. Since the table, which is part of the input, has exponential size in the number of agents, this severely limits the number of agents that can be used in a reduction. These restrictions often disqualify the most natural lower-bound reduction strategies, and currently, there are no new techniques in the literature to provide lower bounds under these restrictions. This issue has led to many missing lower bounds in the literature, with~\cite{RBV23} being one such example.

In response, researchers often restrict general concurrent multi-agent systems in order to avoid this issue of representation, as restricted models often allow for the use of smaller transition tables, which lifts the restriction on the number of agents used in the lower bound. Two major examples of this in the literature are the \emph{iterated boolean game} (iBG)~\cite{iBG} and the turn-based game. In an iBG, the set of collective actions and the set of global states are required to be the same set, meaning that the transition information can be computed through a trivial identity function. In a turn-based game, concurrency is removed, and only one agent chooses a relevant action at a global state, greatly reducing the size of the transition table~\cite{sofiagamesnotes}. Other, more recent examples include the ``bounded-channel"~\cite{RV22} and the ``bounded-concurrency" model~\cite{RV25}.  All of these models are constructed with the same purpose in mind: by sacrificing the general concurrency of a multi-agent system through restricted models, it becomes possible to restrict the size of the transition table, which allows for the establishment of lower bounds. 

While considering restrictions is a valid workaround, this leads to a ``model gap" that devalues the complexity-theoretic results obtained.  Consider~\cite{BBMU15}, which explicitly discusses the issues caused by explicit transition table representation (see Remark 2.2). In that work, an NP upper bound is given for the existence of a ``pure Nash equilibrium" in a concurrent multi-agent system. To provide a matching lower bound, a turn-based system is employed. The intuition there is that since a turn-based system can be naturally considered as a special type of concurrent system, a lower bound on turn-based systems naturally extends to concurrent systems. Yet note that the reduction (Section 5.1.3) \emph{does not} hold for general concurrent systems, as it is the specific use of a turn-based system that allows for a transition table that can be constructed in polynomial time. This is the \emph{model gap}, a common pattern in equilibrium analysis, in which an upper bound is shown for a highly general system, but the lower bound can only be demonstrated for a special restriction of the system that admits a beneficial representation. The model gap obscures the true complexity of equilibrium decision problems by making it unclear whether the lower bound truly applies to the general model, a point we explicitly discuss in Section~\ref{realizabilitydiscussion} for the NP-completeness result of~\cite{BBMU15} . Finally, it should be noted that in~\cite{BBMU15}, the use of this explicit transition table model is motivated by a desire to compare the results in the paper to results from other papers in the literature, as the explicit transition table model is ubiquitous in equilibrium analysis.

This brings us to another salient point. While the use of a monolithic table is common practice in the equilibrium analysis literature~\cite{BBMU15}, it is not common practice at all in computer science as a whole. Rather, the common practice is to construct a small program that computes the transition information, rather than explicitly listing it in a monolithic table structure, cf. \cite{burch1994symbolic}. This design pattern is reflected in the examples given by tools such as~\cite{PRISM,MCMAStool}. Therefore, the use of explicit transition tables also disconnects insights gained from complexity-theoretic analysis from the behavior of practical multi-agent-systems analysis tools.

In this paper, we address these issues by introducing  \emph{circuit-based} multi-agent systems. In contrast with explicit transition tables, which are common in the equilibrium literature, the circuit-based model uses sequential boolean circuits to compute state transition information.  Using circuits instead of tables enables us to describe large transition tables succinctly, regardless of the number of agents.
This addresses the model gap, as evidenced by the lower-bound results presented in the paper. Furthermore, the use of circuit-based transition systems is directly aligned with practical representations of multi-agent systems, allowing for much better application of the insight gained from theoretical results to practical systems. These advantages are reflected in highly influential previous works that study circuit-based models in other settings, such as~\cite{seqcircuits,succinctgraphs}.

The paper proceeds as follows. First, we introduce the \emph{explicit} and \emph{circuit-based} multi-agent systems. As previously mentioned, the explicit system is the system often found in the literature, while the circuit-based system is a new contribution and the focus of this paper. We then move on to introduce broad mathematical notions of strategies and our equilibrium solution concept, the $W$-Nash Equilibrium ($W$-NE)~\cite{RV21,RV22}. The $W$-NE is a slight adaptation of the most popular equilibrium concept in the equilibrium analysis literature and game theory literature in general, the Nash equilibrium, for deterministic systems. Furthermore, we discuss the different ways strategies are represented in the explicit and circuit-based systems. We then analyze the two most important equilibrium decision problems of \emph{realizability}~\cite{PnuRos89a}, which is to determine if an equilibrium exists, and \emph{verification}~\cite{Ros92}, which is to check if an input profile of strategies satisfies the equilibrium condition for both the explicit and circuit-based systems in order to compare their complexity-theoretic properties.  Since one of the major contributions of this paper was the introduction of a new model, the new directions made possible by this model, such as the ability to perform more precise multivariate complexity analysis, see e.g.~\cite{Vardi82}, represent an important aspect of this work. Beyond that, further directions for both the model and the line of work in general are discussed in Section~\ref{Conclusion}.

Although this paper contains several novel technical results, its focus is conceptual. The main purposes of this paper are to demonstrate the \emph{model gap} phenomenon and to advocate for the use of succinctly represented circuit-based systems in the literature. For this reason, the paper is organized to first provide the reader with a high-level overview of our results and discussions of how these results impact the overall literature in Sections~\ref{realizabilitydiscussion} and~\ref{verificationdiscussion}. For readers interested in delving into the technical content, the lower-level details of the proofs and constructions underlying our technical results are provided in Section~\ref{appendix}.

\section{The Explicit and Circuit-Based Models}\label{modelsection}

In this section, we introduce the \emph{explicit} and \emph{circuit-based} multi-agent system models. The components and functions of the two systems are the same; the difference between the two lies in how they are represented.

\begin{defi}[Explicit Multi-Agent System]\label{explicitsystemdef}
    An \emph{explicit multi-agent system} is given by a 6-tuple $\mathbb{G} = \langle V, v_0, \Omega = \{0 \ldots k-1\}, A = \{A_0 \ldots A_{k-1} \}, G = \{ G_0 \ldots G_{k-1} \}, \tau \rangle$ with the following interpretations:
        {\bf (1)} $V$ is a finite set of \emph{global states}. $v_0 \in V$ is the \emph{initial state}.
        {\bf (2)} $\Omega = \{ 0 \ldots k-1 \} $ is the set of $k$ \emph{agents}. Following common convention in computer science, the first agent is Agent $0$, making the last Agent $k-1$.
        {\bf (3)} $A = \{A_0 \ldots A_{k-1}\}$ is the set of \emph{action sets}. The set $A_i$ is associated with Agent $i$. As mentioned before, we assume that $\forall i \in \Omega. |A_i| \geq 2$, as this allows each agent to have a meaningful choice of action.  For notational convenience, we also introduce the set of \emph{decisions} $D = \bigtimes_{i \in \Omega} A_i$.
        {\bf (4)} $G = \{ G_0 \ldots G_{k-1} \}$ is the set of \emph{agent goals}. Each goal $G_i \subseteq V$ is a \emph{reachability} goal, a decision we explicitly motivate in Section~\ref{whyreachability}. The goal $G_i$ is associated with Agent $i$.
        {\bf (5)} $\tau$ is the transition function of type $V \times D \rightarrow V$. Of particular importance to this model is that $\tau$ is represented by an \emph{explicit table} with $|V| \cdot |D|$ rows. 

\end{defi}
Note that the transition table in the explicit multi-agent system has an explicit, well-defined notion of \emph{size}, as it is represented by a table with \emph{exactly} $|V| \cdot |D|$ rows. Since $\forall i \in \Omega. |A_i| \geq 2$, we have $|D| \geq 2^k$, meaning that the size of the transition table is exponential in the number of agents. Although this is ``large", we require $|V| \cdot |D|$ rows in order to give $\tau$ a fixed representation and size, as the lack of this notion of size is the major factor behind the model gap.  As opposed to common practice in the equilibrium analysis literature (see the discussion centered around Remark 2.2 in~\cite{BBMU15}), we deliberately take this size into account in order to provide sharper complexity-theoretic bounds. This point, one of the main themes of this paper, is explored in much greater detail in the technical sections below.

The explicit multi-agent system works as follows. The system starts in state $v_0$. At this state, each agent chooses an action from their action set. Their collective choice, a member of the set $D$ of decisions, along with the current state $v_0$, then determines a new state that the system transitions to as specified by the transition function $\tau$. This process then repeats an infinite number of times. Agents interact with the system in order to fulfill their goal, which is to eventually reach a state that belongs to their reachability goal (recall that each goal $G_i$ is a subset of the set $V$ of states). Therefore, if the system visits a state in $G_i$, Agent $i$'s goal is \emph{satisfied}, otherwise it is not. We now move on to motivate why reachability goals specifically are used in this paper.

\paragraph{The Use of Reachability Goals}\label{whyreachability}

Reachability goals are arguably the simplest type of non-trivial agent goal, and it is exactly for this reason that we elect to use them. In the broader equilibrium-analysis literature, it is often common to use infinite-horizon temporal logic goals such as LTL, e.g.~\cite{iBG}. Recently, it has also become increasingly common to use finite-horizon temporal logics such as LTL\textsubscript{f} and LDL\textsubscript{f}~\cite{multildl}. While these goals are expressive, it is also computationally expensive to reason about them even in non-multi-agent settings~\cite{GV13}. This leads to results in which the sub-problem of reasoning about the goals alone drives the entire complexity of reasoning about equilibria in multi-agent systems; an example is the 2EXPTIME upper bound for realizability in~\cite{GPW17,multildl}, which is completely driven by the complexity of reasoning about LDL\textsubscript{f}. This type of result says little about the complexity of reasoning about the strategic interactions in the multi-agent system. Our goal in this paper is to analyze the effect of input representation on overall complexity for equilibrium decision problems in multi-agent systems. Therefore, we opt to keep the goals, which have their own separate issue of representation~\cite{RV22}, as simple as possible in order to isolate the effect of multi-agent system representation. 

We now move on to define the \emph{circuit-based multi-agent system}. The circuit-based system has components that are directly analogous to the components in the explicit system. The major difference is that some of the components have different representations, with the circuit-based system having components that are represented implicitly through assignments to boolean variables. In order to differentiate between components of the explicit system and components of the circuit-based system that serve the same function but have a different representation, we use the ${\tt typewriter}$ font, e.g. $A$ vs $\tt{A}$.

\begin{defi}[Circuit-Based Model]\label{circuitbasedmodeldef}
    A \emph{circuit-based mutli-agent system} is a 6-tuple $\mathcal{G} = \langle \tt{V}, \tt{v_0}, \Omega = \{0 \ldots k-1 \}, \tt{A} = \{ \tt{A_0} \ldots \tt{A_{k-1}}\}, \tt{G} = \{ \tt{G_0} \ldots \tt{G_{k-1}} \}, \varphi \rangle$ with the following interpretations:
        {\bf (1)} $\tt{V}$ is a set of environment \emph{variables}. The states of the system are then the assignments to these variables, i.e. $2^{\tt{V}}$. This means that the ``full" set of states is represented implicitly as opposed to as an explicit set of states as in the explicit model. $\tt{v_0} \in 2^{\tt{V}}$ is the \emph{initial state}.
        {\bf (2)} $\Omega = \{0 \ldots k-1 \}$ is once again the set of agents. There is no difference in the representation of this component between the explicit and circuit-based models.
        {\bf (3)} $\tt{A} = \{ \tt{A_0} \ldots \tt{A_{k-1}}\}$ is the set of action \emph{variable} sets. Each $\tt{A_i}$ is a set of variables, and, just like the set of states, the explicit representation of the action set for Agent $i$ is given by $2^{\tt{A}_i}$. We again introduce the set  $\tt{D} = \bigtimes_{i \in \Omega} {\tt{A_i}}$ of decisions.
        {\bf (4)} $\tt{G} = \{ \tt{G_0} \ldots \tt{G_{k-1}} \}$ is the set of goals, which are once again reachability goals as motivated in Section~\ref{whyreachability}. Each $\tt{G_i}$ is represented as a \emph{combinational circuit} that computes the indicator function of the corresponding reachability set, i.e. $\tt{G_i}: 2^{\tt{V}} \rightarrow \{0,1\}$, with an output of $1$ implying the state is in the set and $0$ implying the state is not in the set. 
        {\bf (5)} $\varphi: 2^{\tt{V}} \times 2^{\tt{D}} \rightarrow 2^{\tt{V}} $ is the transition function, represented by a combinational circuit. 
\end{defi}

The circuit-based system functions in exactly the same manner as the explicit model. The system starts in state $\tt{v_0}$. Agents pick actions, which amount to an assignment to their set of action variables $\tt{A_i}$. This transitions the system to another state, which goes on ad infinitum. The Agent $i$ interacts with the system with the goal of visiting a state $\tt{v} \in 2^{\tt{V}}$ such that $\tt{G_i}(\tt{v}) = 1$.  The difference between the two systems lies in their representation. Namely, instead of explicit sets (such as the state set $V$), implicit representations using boolean variables are given. For functions, we use combinational circuits instead of explicit tables. The relationship between circuit-based and explicit systems is described in the following observation.

\begin{observation}\label{conversion}
A circuit-based system can be unfolded into an explicit one through the following procedure:
    {\bf (1)} Implicitly represented sets such as $\tt{V}$ and $\tt{A_i}$ are unfolded into exponential-sized explicit representations. Instead of $\tt{V}$, for example, we explicitly construct every element of $2^{\tt V}$.  Note that the set $\Omega$ has the same representation in both systems, and therefore, this step does not apply to $\Omega$.
    {\bf (2)} Functions represented by circuits, such as the transition function $\varphi: 2^{\tt{V}} \times 2^{\tt{D}} \rightarrow 2^{\tt{V}} $, are represented by tables that explicitly list each input/output pair. These table representations of circuits are at most exponential in the size of the original circuit in general~\cite{papacomplexity}. 
\end{observation}

Having defined systems, the main question is then how agents choose their actions. This is done through the use of \emph{strategies}, which we formalize mathematically in the next section.

\section{Strategies, Equilibria and Decision Problems}
In this section, we define \emph{strategies}, which enable agents to choose their actions in a systematic way. In addition to providing a mathematical definition, we describe how strategies are represented as transducers in both the explicit and circuit-based models. We then expand this definition to \emph{equilibria}, and then define the decision problems of \emph{verification} and \emph{realizability}. For ease of presentation, we define general concepts with respect to explicit systems, with the understanding that the circuit-based analog can be easily understood by substituting in the relevant components.

\begin{defi}[Strategy for Agent $i$]\label{detstrat}
A strategy for Agent $i$ is a function $\pi_i : D^* \rightarrow A_i$. Intuitively, this is a function that, given the observed history of the game (represented by an element of $D^*$), returns an action $a_i \in A_i$. 
\end{defi}

Note that by tracking the set of decisions $D$, agent strategies capture the full observable history of the game, including the sequence $V^*$ of vertices visited. A strategy of the form $V^* \rightarrow A_i$ can then be encoded as a strategy of the form $D^* \rightarrow A_i$, but the opposite implication does not hold in general. The set of all Agent $i$ strategies is $\Pi_i$.  This brings us to the definition of a \emph{strategy profile}.

\begin{defi}[Strategy Profile]\label{profiledef}
    A strategy profile $\pi \in \bigtimes_{i \in \Omega} \Pi_i$ is a $k$-tuple of strategies that assigns exactly one strategy to every Agent, and therefore represents a function of type $D^* \rightarrow D$. We also write $\pi = \langle \pi_0 \ldots \pi_{k-1} \rangle$ as a way to unfold a strategy profile in terms of its individual components.
\end{defi}

We express agent strategies as \emph{finite-state transducers}. 
Considering finite-state transducers as a representation for strategies is mathematically sufficient, as argued in works such as~\cite{RV22}.
Just as we have both explicit and circuit-based models of multi-agent systems, we have both explicit and circuit-based models of transducers.

\begin{defi}\label{deterministictransducermodel}
An \emph{explicit transducer} representing a deterministic Agent $i$ strategy in an explicit multi-agent system $\mathbb{G}$ is a 5-tuple $\langle S_i , s^i_0 , D, \gamma^i, O_i \rangle$, where
{\bf (1)} $S_i$ is a finite set of states, and $s^i_0 \in S_i$ is the initial state.
{\bf (2)}$D$, the set of decisions in $\mathbb{G}$, is the input alphabet.
{\bf (3)} $\gamma^i : S_i \times D \rightarrow S_i$ is the deterministic transition function, which is once again represented by an \emph{explicit table}.
{\bf (4)} $O_i : S_i \rightarrow A_i$ is the deterministic output function.
(This type of finite-state transducer is also called a \emph{Moore machine}~\cite{sipser2006}.)
\end{defi}

In this way, the explicit transducer represents a function of type $D^* \rightarrow A_i$, making it a suitable model for a strategy $\pi_i$.  Intuitively, the transducer starts in state $s^i_0$ and transitions upon reading elements in $V$ according to $\gamma^i$. Upon reaching a new state, it outputs an element of $A_i$ based on the output function $O_i$. A more formal description of how finite-state transducers work can be found in~\cite{sipser2006}.

In the same way we constructed the explicit transducer model, we can also introduce the circuit-based transducer model. This construction follows the same logic as the move from the explicit multi-agent system to the circuit-based multi-agent system, and follows the same typographical notation of using the ${\tt typewriter}$ font to signify implicitly represented sets and discrete functions represented by circuits.

\begin{defi}\label{circuittransducermodel} A \emph{circuit-based transducer} representing a deterministic Agent $i$ strategy in an circuit-based multi-agent system $\mathcal{G}$ is a 5-tuple $\langle \Salt_i , \s^i_0 , {\tt D}, \omega^i, \Oalt_i \rangle$, where
{\bf (1)} $\Salt_i $ is a set of state propositions for Agent $i$, inducing the set $2^{\Salt_i}$ of states. 
    {\bf (2)} $\s^i_0 \in 2^{\Salt_i}$ is the initial state. 
    {\bf (3)} The input alphabet is represented by the set ${\tt D}$ of decisions in $\mathcal{G}$, implicitly yielding the explicit input alphabet $2^{\tt D}$.
    {\bf (4)} $\omega^i: 2^{\Salt_i} \times 2^{\tt{D}} \rightarrow 2^{\Salt_i}$ is the transition function, which is once again represented by a \emph{combinational circuit}.
    {\bf (5)} $\Oalt_i : 2^{\Salt_i} \rightarrow 2^{\tt{A_i}}$ is the output function, which is once again represented by a \emph{combinational circuit}. Here ${\tt A_i}$ are Agent $i$ action variables, so $A_i=2^{\tt{A_i}}$.

\end{defi}

Note that a circuit-based transducer can be unfolded into an explicit transducer that is of at most exponential size via the same procedure outlined in Observation~\ref{conversion}. In order to simplify the presentation of this paper, the rest of the definitions in this section are given with regard to the explicit-state system notation $\mathbb{G} = \langle V, v_0, \Omega = \{0 \ldots k-1\}, A = \{A_0 \ldots A_{k-1} \}, G = \{ G_0 \ldots G_{k-1} \}, \tau \rangle$. These definitions can be easily extended to the circuit-based notation, since the functionality of the components in both systems is the same.

Agents progress the execution of the system through their choices of actions. These choices over time result in a \emph{decision sequence}, an element $p \in D^{\omega}$. A decision sequence induces a \emph{trace}, which represents an infinite execution of a system:
\begin{defi}\label{tracedef}
    Given an explicit system $\mathbb{G}$, a trace $t \in V^{\omega}$ of $\mathbb{G}$ is an infinite sequence of states in $\mathbb{G}$ that satisfies two properties: {\bf (1)} The first element of the trace, $t[0]$, is $v_0$ and {\bf (2)} for every two consecutive elements in the trace $t[i], t[i+1]$, $\exists d \in D. \tau(t[i], d ) = t[i+1]$. Given a system $\mathbb{G}$, a decision sequence $p \in D^{\omega}$ induces a unique trace according to {\bf (1)} and {\bf (2)}.

\end{defi}
\noindent
Note that every trace of $\mathbb{G}$ is induced by at least one decision sequence.

The definition for circuit-based systems follows is completely analogous.  Since both the strategies and systems considered in this paper are deterministic, a strategy profile corresponds to a unique trace in the system, which we denote as the \emph{primary trace}.

\begin{defi}[Primary Trace resulting from a Strategy Profile]\label{primary trace}
Consider a multi-agent system $\mathbb{G} = \langle V, v_0, \Omega = \{0 \ldots k-1\}, A = \{A_0 \ldots A_{k-1} \}, G = \{ G_0 \ldots G_{k-1} \}, \tau \rangle$.
Given a strategy profile $\pi$, the decision sequence of $\pi$ is the unique sequence of decisions $s \in D^{\omega}$ that satisfies $s[0] = \pi(\varepsilon)$ (with $\varepsilon$ the empty string) and $s[i] = \pi( s[0] \ldots s[i-1])$. The \emph{primary trace} is then the unique trace $t \in V^{\omega}$ resulting from this decision sequence, i.e. $t[0] = v_0$ and $t[i] = \tau(t[i-1], s[i-1])$ as in Definition~\ref{tracedef}. We denote this trace as $t_{\pi}$.
\end{defi}

 Given a trace $t \in V^\omega$, we define the \emph{winning set} $W_t=\{i\in\Omega~:~ t\models G_i \}$ to be the set of agents whose reachability goals are satisfied by $t$ -- an agent goal $G_i$ is satisfied on $t$ if $t$ contains at least one state $v \in G_i$.  The \emph{losing set} is then defined as $ \Omega / W_t$. The decision problems considered in this paper center around the \emph{Nash equilibrium}, a fundamental concept in the game theory literature.
 


 \begin{defi}[Nash Equilibrium]
Let $\mathbb{G}$ be multi-agent system and $\pi = \langle \pi_0 , \pi_1 \ldots  \pi_{k-1}\rangle$ be a strategy profile. We denote $W_{\pi}=W_{t_\pi}$. The profile $\pi$ is a \emph{Nash equilibrium} if for every $i \in \Omega / W_{\pi}$ we have that for every strategy profile of the form $\pi' = \langle\pi_0 , \pi_1 \ldots \pi'_i \ldots  \pi_{k-1}\rangle$, where $\pi^{'}_i \in \Pi_i$, it is the case that $i \in \Omega / W_{\pi '}$.
\end{defi} 
This definition applies the ``no profitable unilateral defection'' characteristic of the Nash equilibrium~\cite{Osborne1994} to the deterministic setting. Since we are working in a deterministic setting, each agent only has two possible payoffs, corresponding to whether they satisfy their reachability goal.  Thus, the key property of a Nash-Equilibrium profile is the set of agents that meet their goals.

\begin{defi}[$W$-NE]\label{WNE}~\cite{RV21}
    Let $\mathbb{G}$ be a multi-agent system,  $\pi = \langle \pi_0 , \pi_1 \ldots  \pi_{k-1}\rangle$ be a strategy profile, and $W \subseteq \Omega$ be a subset of agents. The strategy profile $\pi$ is a $W$-Nash Equilibrium ($W$-NE) iff it is a Nash equilibrium and $W_{\pi} = W$.
\end{defi}

This brings us to the decision problems of \emph{realizability} and \emph{verification} in the multi-agent equilibrium analysis setting.
\begin{problem*}[Realizability]
  ${}$

    {\bf Input:} A multi-agent system $\mathbb{G}$, a set $W \subseteq \Omega$ of agents.
    
    {\bf Output:} {\bf YES} if there exists a strategy profile $\pi$ that is a $W$-NE in $\mathbb{G}$, {\bf NO} if there is no such a strategy profile.
    
\end{problem*}

\begin{problem*}[Verification]
${}$

    {\bf Input:} A multi-agent system $\mathbb{G}$, a set $W \subseteq \Omega$ of agents, and a strategy profile $\pi = \langle \pi_0 \ldots \pi_{k-1} \rangle $ represented by a tuple of finite-state transducers.
    
    {\bf Output:} {\bf YES} if $\pi$ is a $W$-NE in $\mathbb{G}$, {\bf NO}  otherwise.
    
\end{problem*}

We develop the technical results of this paper in the following sections, which characterize the computational complexity of the realizability and verification decision problems for both the explicit and circuit-based models. In doing so, we pay special attention to how they compare and contrast through discussion sections, paying special attention to the ``model gap" issue introduced in Section~\ref{intro}.

A special type of two-agent system called a \emph{reachability game} plays a central role in our analysis. 

\subsection{Reachability Games}\label{reachabilitygamesection}

In this section, we introduce a special type of multi-agent system called a \emph{reachability game} that is important to our analysis. Unlike the general concurrent multi-agent systems introduced in Section~\ref{modelsection}, a reachability game has two agents and is \emph{turn-based}.

\begin{defi}(Reachability Game)\label{reachabilitygamedef}
    A reachability game $G$ is a 5-tuple $G = \langle V_0 , V_1, v_{in}, E, R \rangle$ with the following interpretations:
    \begin{enumerate}
        \item The set $V$ of states is partitioned into two disjoint sets, $V = V_0 \cup V_1$ and $V_0 \cap V_1 = \emptyset$. These two sets correspond to the two agents, $0$ and $1$. $v_{in} \in V$ is the initial state.
        \item $E \subseteq V \times V$ is a directed edge relation on the set of vertices.
        \item $R \subseteq V$ is the \emph{reachability goal}.
    \end{enumerate}
\end{defi}

 The play proceeds as follows. First, the game starts at $v_{in} \in V$. At states in $v \in V_0$, Agent $0$ gets to choose the successor state by selecting a state $v'$ such that $\langle v, v' \rangle \in E$. At states $v \in V_1$, Agent $1$ chooses the successor state in the same manner.  Because transitions in the system only consider an action choice from a single agent at a time, the game is called \emph{turn-based}. The play continues in this fashion ad infinitum. 
 The game is called a reachability game because the reachability set $R \subseteq V$ is associated with Agent $0$. If the play ever visits a state $r \in R$, then we say Agent $0$ wins the game. Otherwise, if the play never visits $R$, then Agent $1$ wins the game.

Given an arbitrary reachability game, it is either the case that Agent $0$ or Agent $1$ has a \emph{winning strategy}~\cite{McNaughton1993InfiniteGP}. A winning strategy for Agent $0$ is a map $V_0 \rightarrow V$ that chooses successor states from states in $V_0$ in such a way that for all possible successor choices in $V_1$, the set $R$ is eventually reached and Agent $0$'s reachability goal is satisfied. Likewise, a winning strategy for Agent $1$ is the dual of this concept, a mapping $V_1 \rightarrow V$ that ensures that every possible choice from Agent $0$ results in $R$ never being reached. Given a reachability game, determining which of the two agents has a winning strategy is PTIME-complete~\cite{sofiagamesnotes}.

It may also be the case that we are interested in an \emph{uninitialized} reachability game of the form $G = \langle V_0, V_1, E, R \rangle$. This is simply a reachability game without an initial state specification. Here, the algorithmic problem of interest is to partition the state space $V$ into two sets, $\Win_0(G)$ and $\Win_1(G)$. The idea behind this partition is that for states $v_0 \in \Win_0(G)$, Agent $0$ has a winning strategy in  $\langle V_0 , V_1,v_{in} = v_0, E, R \rangle$. The same holds for $\Win_1(G)$, as these are states where Agent $1$ has a winning strategy. It is known that $\Win_0(G) \cup \Win_1(G) = V$ and $\Win_0 \cap \Win_1(G) = \emptyset$~\cite{McNaughton1993InfiniteGP}. Given an uninitialized reachability game $G$, the sets $\Win_0(G)$ and $\Win_1(G)$ can be computed in PTIME.  

Reachability games are also sometimes called \emph{safety games} in the literature.  These games are completely equivalent; the only difference is that instead of Agent $0$ having a reachability goal of reaching $R$, we say that Agent $0$ has a ``safety" goal of staying within $V \setminus R$, i.e., never reaching $R$. These games are the same since they just swap the roles of Agent $0$ and Agent $1$, as it can be seen that Agent $1$ has a safety goal in a reachability game.

\section{The Complexity of Realizability}

\subsection{The Explicit Model}\label{explicitmodelrealizabilitysection}

\subsubsection{Upper Bound}\label{upperboundexplicitrealizability}

In~\cite{RV21}, it was shown that the existence of a $W$-NE in a multi-agent system in a specific setting was equivalent to a nonemptiness query in an appropriately constructed B\"uchi automaton. The setting used in~\cite{RV21}, however, used a modification of the \emph{iterated Boolean game} (iBG)~\cite{iBG} model and was therefore not as general as the ones considered in this paper -- in an iterated Boolean Game, the set of states and the set of decisions are necessarily required to be the same, whereas there is not such restriction on our model. The analysis on iBGs in~\cite{RV21} was then extended to a more general setting closely related to our own in~\cite{RBV23}. We generalize this construction to establish our upper bound.

The first step is to create a B\"uchi automaton $A_W$ that recognizes the set of ``valid" primary traces. Recall that the primary trace of a strategy profile $\pi$, as defined in Definition~\ref{primary trace}, is the unique trace $t \in V^{\omega}$ that results when the agents in $\mathbb{G}$ follow $\pi$ without deviating. The $W$-NE solution concept mandates that if $\pi$ is a $W$-NE, then the agents in $W$ and only the agents in $W$ achieve their reachability goal, and so this is the property we expect to hold of the primary trace. The intuitive idea behind the construction of $A_W$ is that a word $p \in D^{\omega}$ is accepted iff it corresponds to a trace that is a valid primary trace. This automaton $A_W$ is then further refined to create the automaton $A'_W$ that recognizes the set of valid primary traces from which agents not in $W$ cannot profitably deviate from. Therefore, non-emptiness in $A'_W$ is equivalent to the existence of a $W$-NE. The full details of the constructions and correctness of both $A_W$ and $A'_W$ are given in Section~\ref{awappendix}.

This gives us our algorithm to decide $W$-NE realizability, as Section~\ref{awappendix} proves that non-emptiness of $A'_W$ is logically equivalent to the question of $W$-NE realizability. Algorithmically, there are two stages to test $A'_W$ for non-emptiness. The first is solving a series of $|\Omega \setminus W|$ reachability games with $V + |V \times D|$ states in order to identify the ``dangerous'' states that must be pruned to create $A'_W$ from $A_W$ (see Section~\ref{awappendix}). While it might seem that this state space is exponential due to the presence of $D$, note that we are working with the \emph{explicit} model, where the transition function $\tau$ is represented by a table with exactly $|V| \cdot |D|$ rows. This means that the reachability games actually have \emph{polynomial} size with respect to the input, and they can therefore be solved in polynomial time~\cite{McNaughton1993InfiniteGP}. The next component is to test $A'_W$ for nonemptiness. Recall that in order to give each agent agency, we require that $\forall i \in \Omega. |A_i| \geq 2$, meaning that $|D| \geq 2^k$. Therefore, once again, the state space of $A'_W$ is polynomial in the size of the input. A Büchi automaton can be tested for nonemptiness in NLOGSPACE~\cite{VW94}. Therefore, the PTIME complexity of solving the reachability games ultimately dominates the NLOGSPACE complexity of testing $A'_W$ for nonemptiness, and our final result is that $W$-NE realizability can be decided in PTIME.


\begin{thm}\label{explicitrealizabilityupperboundtheorem}
    For the explicit model, the $W$-NE realizability problem is in PTIME.
\end{thm}

This PTIME upper bound merits further comparison with~\cite{BBMU15}, which we make in Section~\ref{realizabilitydiscussion}. In order to complete the technical analysis of the $W$-NE realizability problem, we now provide a matching lower bound. 

\subsubsection{Lower Bound}\label{lowerboundexplicitrealizability}

In order to prove that $W$-NE realizability for explicit multi-agent systems is PTIME-hard, we reduce from the PTIME-hard problem of solving an initialized reachability game $\mathcal{R}$ (see Section~\ref{reachabilitygamesection}). We reduce this PTIME-complete problem to the existence of an $\emptyset$-NE in an appropriately constructed explicit multi-agent system $\mathbb{G}_\mathcal{R}$.  

\begin{lem}\label{explicitrealizabilitylowerboundlemma}
    An $\emptyset$-NE exists in $\mathbb{G}_\mathcal{R}$ iff Agent 1 has a winning strategy in $\mathcal{R}$.
\end{lem}

The details of the construction of $\mathbb{G}_\mathcal{R}$ given $\mathcal{R}$ and the correctness of the reduction are given in Section~\ref{explicitlowerboundrealizabilityappendix}. Using this reduction, we can prove PTIME-hardness and ultimately conclude:

\begin{thm}\label{explicitrealizabilitycomplete}
    For the explicit model, the $W$-NE realizability problem is PTIME-complete.
\end{thm}

\subsection{The Circuit-Based Model}

\subsubsection{Upper Bound}\label{upperboundcircuitrealizability}

As demonstrated in Observation~\ref{conversion}, the circuit-based model represents an encoding of a multi-agent system that is, at best, exponentially more succinct than the explicit model. 

This gives us an EXPTIME upper bound on the $W$-NE realizability problem by following the exact same algorithms and constructions outlined in Section~\ref{awappendix} on the exponentially-sized unfolding of the circuit-based system described in Observation~\ref{conversion}.   The B\"uchi automaton construction in \ref{awappendix}, for example, had a state space $V \times 2^{\Omega}$. This translates into a state space of $2^{\tt V} \times 2^{\Omega}$ in the circuit-based model, noting that $\Omega$ has the same representation in both models. Since this new construction is exponential in the size of the input, it can be tested for nonemptiness in NPSPACE=PSPACE using the same NLOGSPACE result for B\"uchi automaton nonemptiness~\cite{VW94} utilized in Section~\ref{upperboundexplicitrealizability}.  Likewise, the reachability game constructions that had a state space of size $O(|V \times D|)$ in Section~\ref{awappendix} translate into reachability games of size $O(|2^{\tt V} \times 2^{\tt{D}}|)$, noting that $2^{\tt{D}}$ is an explicit representation of the cross-product of the explicit representation of the action sets in the circuit-based model. This gives us exponentially-sized reachability games which can be solved in EXPTIME. As with the explicit model, the complexity of the reachability game dominates the complexity of the non-emptiness queries, placing the problem in EXPTIME.

\begin{thm}
    For the circuit-based model, the $W$-NE realizability problem is in EXPTIME. 
\end{thm}

\subsubsection{Lower Bound}\label{lowerboundcircuitrealizability}

In order to show EXPTIME-hardness, we reduce from the problem of alternating PSPACE Turing Machine acceptance.
\begin{problem*}    
Given an alternating Turing Machine $M$ and a natural number $n \in \mathbb{N}$ written in unary, does $M$ accept the empty tape using at most $n$ cells? 
\end{problem*}

Given an alternating Turing machine $M$, we can construct a circuit-based multi-agent system $\mathcal{G}$ such that:

\begin{lem}\label{circuitbasedrealizabilitylowerboundlemma}
    $M$ accepts the empty tape iff there is no $\emptyset$-NE in $\mathcal{G}$.
\end{lem}

Details about alternating Turing machines along with the full reduction and proof of Lemma~\ref{circuitbasedrealizabilitylowerboundlemma} can be found in Section~\ref{circuitbasedrealizabilitylowerboundappendix}. We note that this proof has a very similar structure to the one in Section~\ref{explicitlowerboundrealizabilityappendix}, as it essentially establishes that solving succinctly-represented reachability games is EXPTIME-hard and uses the same $\emptyset$-NE encoding of a reachability game as a $W$-NE. This gives us the validity of our reduction and our completeness result.

\begin{thm}\label{circuitrealizabiltycomplete}
    For the circuit-based model, the $W$-NE realizability problem is EXPTIME-complete. 
\end{thm}

\subsection{Discussion}\label{realizabilitydiscussion}
On an intuitive level, comparing the PTIME-complete result for the explicit model in Theorem~\ref{explicitrealizabilitycomplete} to the EXPTIME-complete result for the circuit-based model in Theorem~\ref{circuitrealizabiltycomplete} results in the ``standard" exponential jump that one would expect when making a model exponentially more succinct (for example, see~\cite{succinctgraphs}). There are, however, two very important points for discussion.

The first is the aforementioned NP-complete result of~\cite{BBMU15}, which should be compared to the PTIME-complete result of Theorem~\ref{explicitrealizabilitycomplete}. In order to compare the two results, it must be noted that while the problems of ``constrained Nash equilibrium existence" and ``W-NE realizability" are equivalent, the models considered in this paper and~\cite{BBMU15} are different -- in~\cite{BBMU15}, strategies are of the form $V^* \rightarrow A_i$, not $D^* \rightarrow A_i$. The system used for the NP-hard lower bound in~\cite{BBMU15}, however, is \emph{turn-based}, which removes this distinction and, furthermore, creates what we call a ``model-gap".  

We now explain the model-gap issue by introducing the two models that have a ``gap" between them. In~\cite{BBMU15}, multi-agent systems are first presented in a general manner, where transition tables resemble the explicit tables of Definition~\ref{explicitsystemdef}.  It is then noted that these transition tables are unavoidably large in the worst case (Remark 2.2), an observation we also make when noting that the transition table of an explicit system is exponential in the number of agents. The problem of ``constrained Nash equilibrium existence" is then shown to be in NP. In fact, we have proven a PTIME upper bound in our model. This explicit table model is the first model.

The second model is introduced in the hardness proof, presented in Section 5.1.3 of ~\cite{BBMU15}. There, NP-hardness is shown for a system in which one agent decides every action -- although the other agents have their own goals and therefore add complexity to the system, they do not influence the transitions of the system. This means that the exponential blow-up incurred by considering the joint action space is completely avoided, making a polynomial-time reduction possible. But although it seems natural that a hardness result for a system in which only one agent acts should naturally carry over to the intuitively harder problem of considering a system where multiple agents act concurrently, this is not the case from a complexity-theoretic standpoint. Allowing only one agent to act is a special restriction on the system that allows for a more succinct representation of the transition table.  Therefore, from a complexity-theoretic viewpoint, these are two completely different problems, as the input is specified in two different ways. The gap between the general model for the upper bound and the restricted model for the lower bound is what we refer to as the \emph{model gap}, and we argue that this type of reasoning makes it unclear whether the NP-hard result of~\cite{BBMU15} holds for general concurrent systems or just ones that admit a succinct representation. This is why we took care to precisely specify the representation of our models, as this allows for more mathematically precise complexity results. 

This mathematical precision brings us to our second major point. The reductions presented in Lemma~\ref{explicitrealizabilitylowerboundlemma} and Lemma~\ref{circuitbasedrealizabilitylowerboundlemma} both reduce problems to $\emptyset$-NE realizability in two-agent systems. Since the number of agents in these reductions is fixed, this can also be seen as a ``model gap" in the same way as reducing to a system in which only one agent acts at a time, as both models avoid considering a transition function that has to reason about an exponentially-sized joint action set. 
Let us then consider how the reductions of Lemma~\ref{explicitrealizabilitylowerboundlemma} and Lemma~\ref{circuitbasedrealizabilitylowerboundlemma} behave when we insist that the systems we reduce to are general concurrent multi-agent systems, not just two-agent systems. That is, when reducing from a decision problem $L$, we require that an instance $l$ of $L$ of size $n$ be reduced to a system with $f(n)$ agents, where $f(\cdot)$ is an increasing function. Although these additional agents may not influence the system's transitions in a meaningful way, the transition functions must still model their action choices. 

For the circuit-based system, the reduction in Lemma~\ref{circuitbasedrealizabilitylowerboundlemma} can be done with a polynomial $f(\cdot)$. In other words, we can pad the reduction in Lemma~\ref{circuitbasedrealizabilitylowerboundlemma} with a polynomial number of ``useless" agents that do not influence the transitions of the system (but, unlike the aforementioned example from~\cite{BBMU15}, still factor into the representation of the transition function). The addition of an individual useless agent only increases the size of the system linearly, meaning that reducing to a system with a polynomial number of agents (the two agents originally in the reduction plus a polynomial number of useless agents) in Lemma~\ref{circuitbasedrealizabilitylowerboundlemma} still constitutes a valid polynomial time reduction. This is due to two major factors. First, the actions are represented symbolically, and so the representation of the joint action space grows linearly with the addition of each action. Second, the transition function is a circuit that can perform computations, giving us a way to coherently represent the fact that an agent's action should be ``ignored'' in the transition function.

For the explicit system, neither of the aforementioned points holds. Suppose that the number of agents we use in a reduction is $k$. Then, the transition table of the explicit system would be of size $O(2^k)$. Recall that the source of the reduction is an instance $l \in L$ of size $n$. In order for $2^k$ to be polynomial in $n$, a requirement of the LOGSPACE reduction used in Lemma~\ref{explicitrealizabilitylowerboundlemma}, we must have $k \leq \mathrm{polylog}(n)$. In Section~\ref{lowerboundexplicitverification}, we show that this ``polylogarithmic number of agents" restriction plays a significant role in our results on the complexity of the verification problem for explicit systems.  Independent of this paper, however, this restriction represents a major open issue in the literature. No such reduction using a polylogarithmic number of agents exists in the literature, and this is precisely why researchers have chosen to utilize restricted systems that admit smaller transition tables for lower bounds (for a few examples, see~\cite{BBMU15,iBG,RV22,RV25}). The ``\emph{model-gap argument}" then occurs when lower bounds in restricted systems are combined with upper bounds in more general systems, resulting in incorrect completeness statements.

In this way, the reduction for circuit-based systems can be seen as ``more robust" than that for explicit systems, a point that is further detailed in the verification problem discussion, Section~\ref{verificationdiscussion}. This further supports the idea that the circuit-based model should be adopted over explicit systems in the literature, the main point of this paper.


\section{The Complexity of Verification}

The verification problem takes a multi-agent system with $k$ agents, a subset of agents $W \subseteq \Omega$, and a strategy profile specified as $k$ finite-state transducers as input. In this work, the explicit and circuit-based representations of each are ``paired"; when we refer to the explicit model, we refer to an input specified by both an explicit multi-agent system and a strategy profile given by explicit finite-state transducers. The same applies to the circuit-based model.

\subsection{Explicit Model}

\subsubsection{Upper Bound}\label{upperboundexplicitverification}

The verification problem takes an explicit multi-agent system $\mathbb{G} = \langle V, v_0, \Omega = \{0 \ldots k-1\}, A = \{A_0 \ldots A_{k-1} \}, G = \{ G_0 \ldots G_{k-1} \}, \tau \rangle$ (Definition~\ref{explicitsystemdef}), a subset $W \subseteq \Omega$ of agents, and a strategy profile specified by explicit transducers (Definition~\ref{deterministictransducermodel}) as input. In a system with $k$ agents, the strategy profile takes the form of a tuple $\pi = \langle \pi_0 \ldots \pi_{k-1} \rangle$ where each $\pi_i$ is an explicit transducer $\pi_i = \langle S_i , s^i_0 , D, \gamma^i, O_i \rangle$.

Given the $k$ individual agent strategies $\pi_0 \ldots \pi_{k-1}$, the first conceptual step in establishing an upper bound for the verification problem is to analyze how all of these transducers behave together as opposed to individually, as this allows us to analyze how $\pi$ determines the execution of $\mathbb{G}$. In order to do this, we construct the cross product of all $k$ strategies to create a new transducer $\langle S, s_0 , D, \gamma, O \rangle$ that represents $\pi$ (and we therefore refer to this transducer when we reference $\pi$ going forward). Once we establish the details of the construction of $\pi$, we show that the verification problem can be solved through reachability queries in the cross-product construction $\mathbb{G} \times \pi$. This cross-product construction is exponentially large in the size of the input strategies; therefore, we never directly construct it, but instead reason about it on-the-fly. Full details are given in Section~\ref{explicitverificationupperboundappendix}. Since reachability queries can be conducted in NLOGSPACE~\cite{VW94}, this gives us our upper bound.
\begin{thm}
    For the explicit model, the $W$-NE verification problem is in PSPACE.
\end{thm}

\subsubsection{Lower Bound}\label{lowerboundexplicitverification}

In this section, we demonstrate the matching PSPACE-hard lower bound for the restricted family of \emph{turn-based} systems (see~\ref{reachabilitygamesection} for more details). This is a quintessential model gap argument, as in turn-based systems, only one agent chooses an action at a given state, and therefore, the transition function never needs to represent the exponential joint action space of multiple agents choosing actions simultaneously. 

We discuss explicitly why such an argument is used in the verification discussion section, Section~\ref{verificationdiscussion}. Essentially, this is a feature, not a bug, as the problem of establishing a lower bound with the explicit model given has been a recurring issue in the literature for some time (see~\cite{RV22,RBV23} for some discussion). In this paper, we argue for the use of the circuit-based model, and the lack of a ``proper" reduction only bolsters this argument, as we explain after. The problem we reduce from is:
\begin{problem*}
    Given a deterministic Turing Machine $M$ and a natural number $n \in \mathbb{N}$ written in unary, does $M$ accept the empty tape using at most $n$ cells?
\end{problem*}

And as before, the full details of the reduction, which reduces the acceptance of $M$ to the existence of an $\Omega$-NE in an explicit multi-agent system $\mathbb{G}$, can be found in Section~\ref{explicitlowerboundverificationappendix}.

\begin{lem}\label{explicitverficationlowerboundlemma}
    The machine $M$ accepts the empty tape in at most $n$ iff the strategy profile $\pi$ is an $\Omega$-NE in $\mathbb{G}$.
\end{lem}

As mentioned before, the reduction detailed in Section~\ref{explicitlowerboundverificationappendix} only holds for a \emph{turn-based} system $\mathbb{G}$.

\begin{thm}
        For the {\bf turn-based} explicit model, the $W$-NE verification problem is PSPACE-complete.
\end{thm}

The use of the model gap for this result will be further discussed in Section~\ref{verificationdiscussion}. The technical results of this paper now continue with a consideration of the circuit-based model.

\subsection{Circuit-Based Model}

\subsubsection{Upper Bound}\label{upperboundcircuitverification}

Much as in Section~\ref{upperboundcircuitrealizability}, it is possible to ``unpack" the succinct representation of the circuit-based transducers and circuit-based multi-agent system to create equivalent explicit systems in at most exponential time (see Observation~\ref{conversion}). 

Just as before in Section~\ref{explicitverificationupperboundappendix}, we can construct a cross-product of the circuit-based transducers $\pi_0 \ldots \pi_{k-1}$ to create a cross-product transducer $\langle \Salt , \s_0 , {\tt D}, \omega, \Oalt \rangle$. The details of this construction are largely the same as the explicit cross-product construction given in Section~\ref{explicitverificationupperboundappendix}, so they are omitted here. 

The main point ot observe here is that $\pi$ is only \emph{exponential} in the size of the input, even when circuit-based transducers are considered. To see this, consider the explicit unfolding of a circuit-based transducer $\pi_i$ of size $m_i$. Creating an equivalent explicit system results in a system of size $O(2^{m_i})$. Now, taking the cross-product of these explicit unfoldings yields a transducer of size  $O(2^{\sum_{i \in \Omega}m_i})$, which is still only exponential in the size of the input.

The same applies to the circuit-based system  $\mathcal{G} = \langle \tt{V}, \tt{v_0}, \Omega = \{0 \ldots k-1 \}, \tt{A} = \{ \tt{A_0} \ldots \tt{A_{k-1}}\}, \tt{G}   = \{ \tt{G_0} \ldots \tt{G_{k-1}} \}, \varphi \rangle$. Given a circuit-based system $\mathcal{G}$ of size $g$, it is possible to create an explicit system of size at most  $O(2^{g})$. The constructions used in Section~\ref{explicitverificationupperboundappendix}, $\mathcal{G} \times \pi$ and $\mathcal{G}_i \times \pi$, then have size at most  $O(2^{g + \sum_{i \in \Omega}m_i})$, which is still singly-exponential in the size of the original input. 

The key point here is that the exponential blow-up from converting a circuit-based system to an explicit one does not compose with the exponential blow-up that arises when considering an unbounded cross-product. Rather, both blow-ups only result in a single exponential blow-up when considered together. An upper bound algorithm can then be established in exactly the same way as Section~\ref{explicitverificationupperboundappendix}. Due to the fact that the blow-ups do not compose, the analogous constructions $\mathcal{G} \times \pi$ and $\mathcal{G}_i \times \pi$ are both only exponential in the size of the input, meaning that utilizing the same algorithm presented in Section~\ref{explicitverificationupperboundappendix} (which never explicitly constructs the exponentially-large $\mathcal{G} \times \pi$, instead reasoning about it on-the-fly) results in the same PSPACE upper bound. 

\begin{thm}
    For the circuit-based model, the $W$-NE verification problem is in PSPACE.
\end{thm}

\subsubsection{Lower Bound}\label{lowerboundcircuitverification}

As shown in Section~\ref{upperboundexplicitverification}, the circuit-based model somewhat counterintuitively admits the same PSPACE upper bound as the explicit model. The reduction for the circuit-based system, however, uses a single-agent system as opposed to the multi-agent reduction in Lemma~\ref{explicitverficationlowerboundlemma}. This suggests that the complexity of reasoning about circuits is the primary source of the problem's hardness, rather than the complexity of reasoning about the interactions of multiple agents. Therefore, it is worthwhile to note that the lower bound presented in this section shows some parallels with results in the literature that address the complexity of verification with circuits; c.f.~\cite{OBDDcomplexity}. Since the upper bound for the circuit-based-model verification matches the upper bound for explicit-model verification, we reduce from the same problem as in Section~\ref{lowerboundexplicitverification}.

\begin{problem*}
    Given a deterministic Turing Machine $M$ and a natural number $n \in \mathbb{N}$ written in unary, does $M$ accept the empty tape using at most $n$ cells?
\end{problem*}

Although we reduce this to verification of an $\Omega$-NE in a $1$-agent system, this time, however, we \emph{do not} use a model-gap argument (as defined in Section~\ref{realizabilitydiscussion}). This point will be explicitly discussed in Section~\ref{verificationdiscussion}. As before, the full details are given in Section~\ref{circuitlowerboundverificationappendix}. 
\begin{lem}
    The machine $M$ accepts the empty tape in at most $n$ iff the strategy profile $\pi$ is an $\Omega$-NE in the one-agent system $\mathcal{G}$.
\end{lem}

\begin{thm}
For the circuit-based model, the $W$-NE verification problem is PSPACE-complete. 
\end{thm}

\subsection{Discussion}\label{verificationdiscussion}
There are two important points for discussion regarding the verification problem. The first point regards the use of a model-gap argument (see Section~\ref{realizabilitydiscussion}) when establishing the lower bound for the explicit model in Section~\ref{explicitlowerboundverificationappendix}. The second point discusses why the lower bound in Section~\ref{circuitlowerboundverificationappendix} is \emph{not} a model-gap argument.

In Section~\ref{explicitverificationupperboundappendix}, the PSPACE complexity of the algorithm was entirely driven by the fact that an unbounded cross-product of $k$ transducers had to be taken. The natural question then becomes, why is the lower bound of Lemma~\ref{explicitverficationlowerboundlemma} only given for turn-based systems? Why is there no tight complexity result for the general concurrent explicit system? The answer is once again the model-gap, which illustrates the difficulty in establishing lower bounds for explicit systems. 

Suppose that we attempt to show a reduction from a PSPACE-hard problem $L$ to the $W$-NE verification problem for general concurrent explicit systems. If we use a number of agents that is linear in the size of the instance problem $L$, then we run into a problem with the size of the transition table -- namely, it is exponential w.r.t $L$. Recall that the PSPACE upper bound of Section~\ref{explicitverificationupperboundappendix} was driven by the fact that an unbounded cross-product had to be taken. Then, if the number of agents is fixed, this cross-product is no longer unbounded, and the algorithm of Section~\ref{explicitverificationupperboundappendix} runs in PTIME. This means that attempting to use a fixed number of agents in a PSPACE-hard reduction is equivalent to attempting to prove that PSPACE = PTIME.  The only option left is to use a sublinear number of agents, as discussed in Section~\ref{realizabilitydiscussion}.

A reduction that uses a logarithmic number of agents, for example, would be a perfect candidate.  As we noted earlier, however, such a reduction cannot be found in the literature, and the lack of this type of reduction is what motivated researchers to consider restricted models in the first place. In order to better understand why constructing such a reduction is difficult, let us see what happens when we try to modify the reduction in Lemma~\ref{explicitverficationlowerboundlemma} to use a logarithmic, as opposed to a linear, number of agents.

In Lemma~\ref{explicitverficationlowerboundlemma}, it was crucial to our reasoning that each agent only needed to track a subID of a fixed length of at most $3$. That bound was made possible by the fact that there were $n$ cells and $n$ agents. If we had a logarithmic number of agents, each agent would need to keep track of at least $O(\frac{n}{\log(n)})$ cells in order for the whole tape to be covered. The number of strings of length $O(\frac{n}{\log(n)})$ characters cannot be bounded by a polynomial in $n$ of fixed degree as $2^{ \frac{n}{\log(n)}}$ is superpolynomial, so the states spaces $S_i$ of the transducers $\pi_i$ could not be constructed in polynomial time. Therefore, modifying the reduction in Lemma~\ref{explicitverficationlowerboundlemma} to consider a logarithmic number of agents would require a novel non-trivial technique. 

Now, consider the reduction in Section~\ref{circuitlowerboundverificationappendix}. While we use only one agent in the reduction, it is possible to ``pad" this reduction to use $n$ agents as opposed to $1$. This can be seen by noting that adding a single ``useless" agent to a circuit-based system increases the size of the system linearly, as argued in Section~\ref{realizabilitydiscussion}. Therefore, when considering lower bounds for the verification problem, the circuit-based system does not introduce the same difficulties inherent to the explicit system and, therefore, the literature as a whole. As in Section~\ref{realizabilitydiscussion}, we take this as further evidence for the adoption of the circuit-based model over the explicit one in the literature.

\section{Details. Constructions, and Proofs}\label{appendix}

The previous sections focused on providing a high-level overview of the issues that arise when explicit systems are considered, focusing on the model gap. As a result, many lower-level details, constructions, and proofs were omitted to keep the focus on the conceptual narrative. In this section, we provide the lower-level details, constructions, and proofs omitted in the previous sections. When a lemma or theorem from a previous section is re-stated in this section, we use the previous numbering in order to tie the result back to the narrative text. New theorems, lemmas, and definitions are given their own numbering. 

\subsection{Constructions and Correctness of $A_W$ and $A'_W$}\label{awappendix}

This section corresponds to Section~\ref{upperboundexplicitrealizability}.

For a multi-agent system $\mathbb{G} = \langle V, v_0, \Omega = \{0 \ldots k-1\}, A = \{A_0 \ldots A_{k-1} \}, G = \{ G_0 \ldots G_{k-1} \}, \tau \rangle$,
the B\"uchi automaton $A_W$ is defined by the 5-tuple $\langle V \times 2^{\Omega}, \langle v_0, W \rangle, D, \beta, F \rangle$ where 
\begin{enumerate}
    \item $V \times 2^{\Omega}$ is the set of states, with $\langle v_0 , W \rangle$ the initial state.
    \item The set $D = \bigtimes_{i \in \Omega} A_i$
    of decisions in $\mathbb{G}$, is the input alphabet
    \item $\beta: V \times 2^{\Omega} \times D \rightarrow V \times 2^{\Omega}$ is the transition function, defined as follows. For a state $\langle v, U \rangle \in V \times 2^\Omega$ and a decision $d \in D$, we have that $\beta(\langle v,s \rangle , d)$ is given by $\langle \tau(v,d), U' \rangle$, where $\tau$ is the transition function of $\mathbb{G}$. It now remains to describe how $U$ updates to $U'$. Intuitively, $U \in 2^{\Omega}$ tracks all of the agent goals for agents in $W$ to make sure that all goals for agents in $W$ and only the agents in $W$ are eventually satisfied. 
    Thus, if the transition moves to a state $\tau(v,d) \in G_j$ for some $j \not \in W$, then we say that the entire transition $\beta(\langle v,U \rangle, d)$ is undefined; upon reading $d \in D$ the automaton gets stuck and rejects. This represents the case where an agent outside of $W$ meets their goal, corresponding to a violation of the $W$-NE requirements and thus resulting in rejection by $A_W$.
    On the other hand, if the transition moves to a state $ \tau(v,d) \in G_i$ for some $i \in W$ and $ i \in U$, then Agent $i$'s goal has been met, so we can set $i \not \in U'$; that is Agent $i$'s goal has been removed from the list of outstanding goals in $W$ to satisfy. It is possible for the goals of multiple agents to be met by one transition, in which case these multiple agents are removed from $U$ simultaneously. 
    \item $F$, the set of final states, is given by the set $V \times \{\emptyset \}$. These states represent the set of states that can only be reached after all agents in $W$ have fulfilled their reachability goals. If the states in $F$ are visited infinitely often without $A_W$ getting stuck, this corresponds to a valid primary trace.
\end{enumerate}

A word $p \in D^{\omega}$ is then accepted by $A_W$ iff it corresponds to a decision sequence that corresponds to a valid primary trace $t \in V^{\omega}$ in $\mathbb{G}$. This can be seen intuitively by noting that the automaton $A_W$ starts in $\langle v_0, W \rangle$ and must satisfy the goals for each $i \in W$ before moving to a final state, and further noting that if a goal for an agent $j \not \in W$ is ever satisfied, then the automaton rejects. 

\begin{lem}
 Given an explicit multi-agent system $\mathbb{G}$ and a set of agents $W \subset \Omega$, a decision sequence $p \in D^{\omega}$ induces a valid primary trace $t \in V^{\omega}$ -- that is, one that satisfies all of the goals $\{G_i : i \in W\}$ and none of the goals $\{G_j : j \not \in W \}$) -- iff $p \in \mathcal{L}(A_W)$.
\end{lem}
\begin{proof}
    Suppose $p \in D^{\omega}$ is accepted by $A_W$. Then, it must be the case that no goal $\{G_j : j \not \in W \}$ is satisfied by $p$; otherwise $A_W$ would get stuck and reject. Furthermore, it must be the case that all goals $\{G_i : i \in W\}$ are satisfied by $p$, as the initial state of $A_W$ is $\langle v_0, W \rangle$
    and the only way to discharge an agent $i$ from the $W$-state-component is to visit $G_i$. Since the set of final states $V \times \{\emptyset\}$ is eventually reached by accepting words, reaching such a final state is equivalent to all goals $\{G_i : i \in W\}$ already having been satisfied. Therefore, $p$ corresponds to a trace $t \in V^{\omega}$ that is a valid primary trace.

    Now suppose that $p \in D^{*}$ corresponds to a valid primary trace $t \in V^{\omega}$ in $\mathbb{G}$. This means that in $\mathbb{G}$, $t$ eventually visits all of the sets $\{G_i : i \in W\}$ and none of the sets $\{G_j : j \not \in W\}$. By essentially the same logic as above, this means that $A_W$ accepts $p$, as visiting all of $\{G_i : i \in W\}$ means that $A_W$ reaches its set of final states, and avoiding all $\{G_j : j \not \in W\}$ means it never gets stuck and rejects. 
\end{proof}


The validity of the primary trace is not the only condition of the $W$-NE. It must also be checked that for every agent $j \not \in W$, it is not the case that they can unilaterally deviate from $\pi$ in order to satisfy their goal $G_j$. This condition can be checked by constructing an uninitialized reachability game (see Section~\ref{reachabilitygamesection}) $\mathbb{G}_j = \langle V \times D, V, E, G_j \rangle$ with an edge relation $E$ defined as follows: for $v \in V$, $\langle v, \langle v, d \rangle \rangle \in E$ for all decisions $d \in D$. For $\langle v, d \rangle \in V \times D$, we have $\langle \langle v,d \rangle, v' \rangle \in E$ iff there is an action $a_j \in A_j$ such that $\tau(v , d[a_j]) = v'$, where $d[a_j]$ refers to the decision $d \in D$ with its $j$-th component replaced by $a_j \in A_j$.

Intuitively, this game represents a deviation scenario from Agent $j$, who is represented by Agent $0$ in the reachability game. At state $v \in V$, Agent $1$, who represents a coalition of the other agents, announces a collective decision $d$ for all of the agents. Agent $0$, who plays the role of Agent $j$ in the reachability game, is then allowed to deviate from this decision by solely changing the $A_j$ component and choosing a different action, leading to a potentially new successor state. The goal of Agent $0$ is to eventually reach $G_j$, the reachability goal of Agent $j$ in $\mathbb{G}$.

The states $\Win_0(\mathbb{G}_j)$ then represent the ``dangerous" states from which Agent $j$ can force play to visit $G_j$. Visiting states in $\Win_0(\mathbb{G}_j)$ leads to violation of the $W$-NE condition, as once such a state is visited, Agent $j$ has a beneficial deviation. We now construct a B\"uchi automaton $A'_W = \langle V' \times 2^\Omega, \langle v_0 , W \rangle, D, \beta', F \rangle$, which is the B\"uchi automaton $A_W$ with both $V$ and $\beta$ restricted. First, we prune out from $V$ states belonging to $\Win_0(\mathbb{G}_j)$ for agents $j \not \in W$, yielding $V'$. For $\beta'$, if transition moves to a state outside of $V'$ or if $A'_W$ is in a state $v$ and reads $d \in D$ such that $\langle v, d\rangle  \in \Win_0(\mathbb{G}_j)$ for some $j \not \in W$, then the transition is undefined, and the automaton gets stuck, so it rejects. The main result of our upper bound algorithm is then:
\begin{thm}
    Nonemptiness of $A'_W$ is equivalent to the existence of a $W$-NE in an explicit multi-agent system $\mathbb{G}$.  
\end{thm}

\begin{proof}
Before moving on to the proof, we note that analogous results have been proven in~\cite{RV21,RBV23}, but the setting used in this paper is more general than~\cite{RV21} and not directly comparable to~\cite{RBV23}. We now develop a proof tailored to our setting, making a specific effort to avoid the highly technical tree-automata machinery of~\cite{RV21}. 

 Suppose that $A'_W$ is nonempty. Then, it accepts at least one word $p \in D^{\omega}$.  By the construction of $A'_W$, this word induces a trace $t \in V^{\omega}$. For convenience, we interleave both the trace and the word to create a sequence of pairs $\sigma = \langle v_0 , d_0 \rangle \ldots \in ( V \times D)^{\omega}$ where $v[i] = t[i]$ and $d[i] = p[i]$.

    By the construction of $A'_W$, $\sigma$ does not contain a pair $\langle v, d \rangle$ such that $\langle v, d \rangle \in \Win_0(G_j)$ for some agent $j \not \in W$ -- if it did, then $A'_W$ would get stuck while reading $p$ and reject.  Furthermore, we also know that $t$ visits every set $\{ G_i : i \in W \}$ and none of the sets $\{ G_j : j \not\in W \}$.

    A strategy profile $\pi$ is a $W$-NE (Definition~\ref{WNE}) iff it satisfies two conditions~\cite{RV21}.
    \begin{enumerate}
        \item If no agents deviate from $\pi$, then only the agents in $W$ satisfy their goal.
        \item All agents not in $W$ are not capable of unilaterally deviating to satisfy their goal.
    \end{enumerate}

    We now use $\sigma$ to construct a $W$-NE strategy profile $\pi$. The idea is that $p$ serves as the output for $\pi$ when no deviations occur. When a deviation does occur from Agent $j$, $\pi$ simulates the winning strategy for Agent $1$ in $G_j$. Furthermore, it should be noted that since agents observe the history $D^*$ through their strategies, once an Agent $j$ deviates, it is immediately clear to the other agents that Agent $j$ has deviated.

    The first point, that $p$ is the output of $\pi$ when no deviations occur, is relatively straightforward. It should be noted that this output is achievable by finite-state transducers, as if a B\"uchi automaton is nonempty, then it accepts at least one ultimately periodic word of the form $\alpha \cdot \beta^\omega$ where both $\alpha$ and $\beta$ are finite. 

    For the second point concerning deviations, consider that every element of $\sigma$ belongs to $\Win_1(G_j)$ for all agents $j \not \in W$ -- this can be seen by noting that they necessarily do not belong to $\Win_0(G_j)$ and that $\Win_1(G_j)$ and $\Win_0(G_j)$ partition the entire state space of $G_j$.

    This means that for every pair $\langle v,d \rangle \in \sigma$, $\langle v,d \rangle \in \Win_1(G_j)$. Note that these states are controlled by Agent $0$ (who represents Agent $j$) in $G_j$. Therefore, since $\langle v,d \rangle \in \Win_1(G_j)$, it must be the case that all transitions from $\langle v,d \rangle$ (representing the possible deviations from Agent $j$) transition to states in $\Win_1(G_j)$ as well -- otherwise, a transition to $\Win_0(G_j)$ would be possible and Agent $0$ would have a winning strategy from that point on, not Agent $1$. 

    This gives us the construction of our $W$-NE $\pi$. When no deviations occur, $\pi$ outputs $p$. When a deviation does occur, $\pi$ simulates the winning strategy in $G_j$ to ensure that the deviating Agent $j$ cannot meet their goal. Since $p$ can be assumed to be ultimately periodic and the games $G_j$ are reachability games that admit memoryless winning strategies, $\pi$ can be implemented through the use of finite-state transducers.

    Now suppose that a $W$-NE $\pi$ exists in $\mathbb{G}$. By taking its output when no deviations occur, we get a word $p \in D^{\omega}$. This $p$ must be accepted by $A'_W$ by the exact same logic as above. It must be the case that $p$ induces a trace that satisfies only the goals for agents in $W$, satisfying the first requirement, and it cannot correspond to a state-decision pair $\sigma$ that contains an element belonging to $\Win_0(G_j)$ for some $j \not \in W$ -- otherwise, Agent $j$ could profitably deviate from $\pi$ by playing the winning strategy in $G_j$. These are the two conditions needed for $A'_W$ to accept, finishing the proof.

\end{proof}

\subsection{Explicit Realizability Lower Bound Reduction}\label{explicitlowerboundrealizabilityappendix}

This section corresponds to Section~\ref{lowerboundexplicitrealizability}.

We reduce the problem of solving an initialized reachability game $\mathcal{R} = V_0, V_1, v_{in}, E , R \rangle $  to the existence of an $\emptyset$-NE in the following explicit multi-agent system $\mathbb{G}_{\mathcal{R}} = \langle V, v_{in}, \{0,1\}, A = \{ V, V \}, G = \{ R \cup \hat{1},\emptyset\}, \tau \rangle$ with the following interpretations.
\begin{enumerate}
    \item $V= V_0 \cup V_1 \cup \hat{0} \cup \hat{1}$, the set of states, is the same as the set of states of $\mathcal{R}$, with the addition of two new states, $\hat{0}$ and $\hat{1}$. The initial state $v_{in}$ remains unchanged.
    \item The set of agents, $\{0,1\}$, as in $\mathcal{R}$.
    \item The action sets correspond to the set of states.
    \item The reachability goals are $R \cup \hat{1}$ for Agent $0$ and $\emptyset$ for Agent $1$ (that is, Agent $1$ has an unsatisfiable goal). 
    \item $\tau$, the transition function, is defined as follows. To compute the transition from a state $v$ and a decision $\langle v_0,v_1 \rangle$, it must first be determined whether $v \in V_0$ or $V_1$. If $v \in V_0$, then $v_1$ is ignored. If $\langle v, v_0 \rangle \in E$, then transition moves to $v_0$; otherwise, transition moves to $\hat{0}$, signifying an illegal move made by Agent $0$. 
    Analogously, if $v \in V_1$, then $v_0$ is ignored. If $\langle v, v_1 \rangle \in E$, then transition moves to $v_1$; otherwise, transition moves to $\hat{1}$, signifying an illegal move made by Agent $1$.  The states $\hat{0}$ and $\hat{1}$ are sink states, meaning they self-transition on every input $d \in D$.
\end{enumerate}

Note that although this multi-agent system involves a polynomial blow-up from the original reachability game specification (especially considering the size of $\tau$), the reduction can easily be seen to be computable in LOGSPACE, as the new components are obtained through the computation of trivial functions with respect to the input reachability game. This is readily observable, for example, with the transition function $\tau$, as it is essentially just a polynomially-sized padded representation of the edge set $E$. Therefore, this is a valid LOGSPACE reduction.


Consider now what an $\emptyset$-NE looks like in this system.
It would be a strategy profile $\pi=\langle \pi_0, \pi_1 \rangle$ such that (1) neither Agent 0 nor Agent 1 reach their reachability goals on the primary trace for $\pi$,  and (2) neither Agent 0 nor Agent 1 can reach their reachability goals by deviating from their respective strategies, $\pi_0$ and $\pi_1$. Since Agent 1 has an empty (and therefore, unsatisfiable) goal, it is not possible for him to have a profitable deviation. This leaves Agent 0, and so an $\emptyset$-NE exists in $\mathbb{G}_{\mathcal{R}}$ iff Agent 0 does not have a profitable deviation from $\pi$.

\begin{manuallemma}{4.2}
    An $\emptyset$-NE exists in $\mathbb{G}_{\mathcal{R}}$ iff Agent 1 has a winning strategy in $\mathcal{R}$.
\end{manuallemma}

\begin{proof}
    ($\rightarrow$) Let $\pi = \langle \pi_0, \pi_1 \rangle$ be an $\emptyset$-NE in $\mathbb{G}_{\mathcal{R}}$. Since $\pi$ is a $W$-NE, this means that for all Agent 0 strategies $\pi'_0 \in \Pi_0$, the primary trace resulting from $\langle \pi'_0, \pi_1 \rangle$ does not visit the Agent 0 reachability goal set $R \cup \hat{1}$. 
    We now claim that $\pi_1$ is a winning strategy in $\mathcal{R}$ This is a slight abuse of notation as strategies in $\mathcal{R}$ and $\mathbb{G}_{\mathcal{R}}$ are technically of different types, as Agent 1 strategies in $\mathbb{G}_{\mathcal{R}}$ must suggest actions at all states in $\mathbb{G}_{\mathcal{R}}$, not just the ones controlled by Agent 1 in $\mathcal{R}$. We bridge this gap by declaring that the output of $\pi_1$  at states $v \in V_0$ is arbitrary, as their choice does not affect the execution of the game $\mathcal{R}$.

    Suppose that $\pi_1$ was not a winning strategy in $\mathcal{R}$, and that there was some Agent 0 strategy $\pi'_0$ in $\mathcal{R}$ such that $\langle \pi'_0 ,\pi_1 \rangle$ was winning for Agent 0 in $\mathcal{R}$. Then, by following this strategy exactly in $\mathbb{G}_{\mathcal{R}}$, $\langle \pi'_0 ,\pi_1 \rangle$ would visit $R \cup \hat{1}$ and satisfy Agent 0's reachability goal. This is due to the fact that $\pi'_0$ would not suggest an illegal move, as illegal moves are not possible in $\mathcal{R}$, and, other than the illegal move exceptions, the transitions in $\mathcal{R}$ are the same as $\mathbb{G}_{\mathcal{R}}$. This is a contradiction, and so $\pi_1$ must be a winning strategy in $\mathcal{R}$.

    ($\leftarrow$) Suppose that Agent 1 has a winning strategy $\pi_1$ in $\mathcal{R}$. Then, by extending this to an Agent 1 strategy $\pi_1$ in $\mathbb{G}_{\mathcal{R}}$ (declaring arbitrary outputs at states $v \in V_0$ as discussed before), the strategy profile $\pi = \langle \pi_0, \pi_1 \rangle$ is a $W$-NE for all Agent 0 strategies $\pi_0 \in \Pi_0$. The reasoning is largely the same: if it were not the case and Agent 0 had a profitable deviation, then this profitable deviation would disprove the assumption that $\pi_1$ was a winning strategy in $\mathcal{R}$.

    This gives us both directions of the proof, giving us the validity of our reduction. It is also important to note that if no winning strategy for Agent 1 exists in $\mathcal{R}$, then one must exist for Agent 0 (see Section~\ref{reachabilitygamesection} or~\cite{McNaughton1993InfiniteGP}), meaning that if an $\emptyset$-NE does not exist in $\mathbb{G}_{\mathcal{R}}$, then Agent 0 has a winning strategy. 
\end{proof}

This reduction then gives us

\begin{manualtheorem}{4.3}
    For the explicit model, the $W$-NE realizability problem is PTIME-complete.
\end{manualtheorem}

\subsection{Circuit-Based Realizability Background and Lower Bound Reduction}\label{circuitbasedrealizabilitylowerboundappendix}

This section corresponds to Section~\ref{lowerboundcircuitrealizability}.

For the purposes of our reduction, we only provide a brief overview of how alternating machines function; readers looking for further detail about alternation and alternating Turing machines are referred to~\cite{sipser2006}.

\begin{defi}\label{altmachine}
    
An alternating Turing Machine $M$ is a $5$-tuple $M = \langle R, \Gamma , r_0, \Delta,\kappa \rangle$ with the following interpretations: 
\begin{enumerate}
\item  $R$ is the set of states. $r_0 \in R$ is the initial state.
    \item $\Gamma$ is the alphabet. WLOG, we assume it contains an empty symbol $\bot$.
    \item $\Delta : R \times \Gamma \rightarrow \mathcal{P}(R \times \Gamma \times \{\rightarrow,\leftarrow \} )$ 
        is the transition function. The set $\{\rightarrow,\leftarrow \}$ represents the left and right directions -- how the head moves over the tape in the Turing machine. Since this is an alternating Turing machine, the range of the transition function is a powerset representing all possible sets of transitions available from a single state-character pair. For notational convenience, we assume that the range of $\Delta$ on an arbitrary input has at most size two, i.e., there are only two possible successors from each state-character pair. This assumption does not change the complexity-theoretic properties of the decision problem from which we wish to reduce~\cite{alur2002alternating,papacomplexity}. 
    \item $\kappa: R \rightarrow \{ \mathrm{accept}, \mathrm{reject}, \vee, \wedge, -\}$ is the state labeling function that associates each state with its type. A state can be accepting, rejecting, existential ($\vee$), universal ($\wedge$), or deterministic ($-$).
\end{enumerate}

\end{defi}

Intuitively, an alternating Turing Machine generalizes the concept of a nondeterministic Turing Machine. In a nondeterministic machine, it is possible that there are multiple valid transitions from a single configuration. At these ``existential" states,  the machine may ``choose" the outcome that best leads to acceptance. Alternation extends this notion by adding universal states alongside these existential states. Intuitively, the universal states are the dual of existential states, as from a universal state, it must be demonstrated that acceptance can be ensured from \emph{all} valid transitions. This is often modeled through a two-agent game in which one agent controls the transitions at existential states and is incentivized to witness a computation that ends in acceptance, and the other agent controls the universal states and is incentivized to not witness an accepting computation. We now introduce the problem we wish to reduce from.

\begin{problem*}
    Given an alternating Turing machine $M$ and a natural number $n \in \mathbb{N}$ written in unary, does $M$ accept the empty using at most $n$ cells?
\end{problem*}

Note that if the machine $M$ ever moves ``out-of-bounds" and reads a cell outside of the $n$ allowed, the computation is automatically considered to be rejecting. This problem is complete for the class APSPACE, which is equivalent to the class EXPTIME~\cite{alternation} (and therefore, the problem is EXPTIME-complete). Note that there is a requirement that $M$ accept the tape using at most $n$ cells, meaning that $M$ is ``space-bounded". A crucial tool for reasoning about space-bounded machines is the concept of a machine ID, a tool used to capture the state of a machine $M$ at a specific point in time during a computation~\cite{sipser2006}. This definition applies to all machines that are space-bounded, whether they are alternating, non-deterministic, or deterministic. 

\begin{defi}\label{machineIDdef}
Given a space-bounded Turing machine $M$ with a state-space of $R$ and an alphabet $\Gamma$, an ID is an $n$-character 
string over the alphabet $(\Gamma \cup \{\bot\}) \times (R \cup \{\bot\})$  that captures the complete contents of the tape of $M$, a notion made sensible by the fact that $M$ can only ever use at most $n$ cells. The $ i$-th character of an ID describes the contents of the cell $i$, which may either contain a character from $\Gamma$ or be empty (represented by the new symbol $\bot$). Furthermore, the head may also be over cell $i$, which is reflected by the cross-product with $(R \cup \{\bot\})$, which either specifies the state of the machine or indicates that the head is not over cell $i$ through the $\bot$ character. A valid ID only has the head in one position.
    
\end{defi}

\begin{example}\label{IDexample}
    Consider a machine space-bounded machine $M$ with a binary alphabet, i.e. $\{a,b\}$ and a state space of $R$. Furthermore, consider that $n = 5$, meaning that the machine can only use  $5$ cells. An example of an ID would be $\langle \bot, \bot \rangle, \langle a, \bot \rangle, \langle b, r \rangle, \langle a, \bot \rangle, \langle \bot, \bot \rangle$, which conveys the following information :    
    \begin{enumerate}
        \item The left-most cell (cell $0$) and the right most cell (cell $4$) are both empty.
        \item Cell $1$ reads $a$, cell $2$ reads $b$, and cell $3$ reads $a$.
        \item The head of $M$ is currently positioned over cell $2$ and is in state $r \in R$.
    \end{enumerate}
    As a matter of notational convenience, it is useful to omit the second component of the characters of the ID when they are empty (represented by $\bot$). Therefore, this ID could be more easily written as $ \bot, a, \langle b, r \rangle,  a,  \bot$
\end{example}

Given an alternating Turing machine $M$ and a natural number $n$, we construct the two-agent circuit-based system $\mathcal{G} = \langle {\tt V} , {\tt v_0}, \Omega = \{0,1\} , {\tt A} = \{ {\tt A_0 , A_1} \} ,\G = \{{\tt G_0}, \emptyset\}, \varphi \rangle$. In order to make the reduction easier to parse, we introduce the precise definition of these components on-the-fly while detailing the validity of the reduction.

\begin{manuallemma}{4.5}
    $M$ accepts the empty tape iff there is no $\emptyset$-NE in $\mathcal{G}$.
\end{manuallemma}

\begin{proof}
    As mentioned previously, alternating Turing machines are often studied through the lens of a two-agent reachability game (as in Definition~\ref{reachabilitygamedef}), in which an ``existential" agent, incentivized to demonstrate acceptance by $M$, is pitted against a ``universal" agent, incentivized to demonstrate rejection by $M$. This pattern is reflected in the construction of the two-agent system $\mathcal{G}$. The high-level idea behind this reduction is that the two agents play a reachability game over the space of machine IDs of $M$ (see Definition~\ref{machineIDdef}). 
    
    Given the inputs $M$ and $n$, the number of possible machine IDs is given by $(|(\Gamma \cup \{\bot\})| \cdot |(R \cup \{\bot\})|)^{n}$. Recall that $2^{\tt V}$ is the explicit state set of $\mathcal{G}$. We intend this set to encode the set of all machine IDs 
    The size of ${\tt V}$, an implicit representation of the space of machine IDs, is therefore the logarithm of $(|(\Gamma \cup \{\bot\})| \cdot |(R \cup \{\bot\})|)^{n}$, which is polynomial in the size of the input of $M$ and $n$. The initial state, therefore, corresponds to the initial configuration of the tape, which consists of all empty cells and a head in state $r_0$ positioned over the leftmost cell.
    
    In this system, Agent $0$ assumes the role of the existential agent, and Agent $1$ assumes the role of the universal agent. Agent $0$'s reachability goal, therefore, consists of the set of machine IDs in which $M$ is a final state. The circuit ${\tt G_0}$ is then simply a gadget that recognizes such machine IDs, alongside a circuit-based implementation of $\kappa$, which tells which machine states are accepting, and can therefore be constructed in polynomial time. As in the proof of Lemma~\ref{explicitrealizabilitylowerboundlemma}, Agent $1$'s goal is empty.   

    The explicit set of actions for each agent is given by $R \times \Gamma \times \{\rightarrow, \leftarrow\}$, which corresponds to the possible movements and actions of the head. This set has polynomial size with respect to the input, meaning that the representations of ${\tt A}_0$ and ${\tt A}_1$ in the input can be constructed in polynomial time. The game proceeds by each Agent suggesting a possible transition of the head, which includes the state it transitions to ($R$), the character it writes ($\Gamma$), and the direction it moves in ($\{\rightarrow, \leftarrow \}$). 

    These actions are evaluated by the transition function $\varphi$. Note that although both agents select an action,  only one agent's action is considered at a time; if the state of the machine $M$ is existential or deterministic in the current ID, then Agent $0$'s action is considered; otherwise, at universal states, Agent $1$'s action is considered. The transition function $\varphi$ determines which Agent is acting based on the state in $R$ and $\kappa$, and determines its legality based on $\Delta$. If either Agent chooses a transition that moves the head outside of the bounds of the $n$ cells considered, the system moves to a new non-accepting sink state. If Agent $0$ chooses a transition that is not allowed by $\Delta$, then the system immediately moves to the new non-accepting sink state. If Agent $1$ chooses a similarly illegal transition, the system immediately moves to an arbitrary accepting configuration.  This reflects the fact that Agent $0$ is looking to demonstrate an accepting computation and Agent $1$ is looking to demonstrate a rejecting one, which is why the ``penalties" associated with each agent choosing an illegal transition are different -- essentially, this is the same ``illegal-move" mechanism created by the states $\hat{1}$ and $\hat{0}$ in Section~\ref{explicitlowerboundrealizabilityappendix}. Otherwise, if there are no penalties to apply and the head stays in bounds, the transition updates the state based on the local information around the state. Since the update only happens at the head location and one adjacent cell, the circuit of the transition function $\varphi$ can be constructed in polynomial time. This can be seen by noting that the size of the relevant domain of the function that $\varphi$ computes is no larger than $|R \times \Gamma|^3 \times |R \times \Gamma \times \{\rightarrow, \leftarrow\}|$ (the contents of three cells, including the head and the at most two cells that surround it, and an agent action), which is polynomial in the size of the input. 

As mentioned at the start of the proof, acceptance in $M$ can be analyzed through the lens of a reachability game over the space of IDs of $M$. In Lemma~\ref{explicitrealizabilitylowerboundlemma}, we showed that the problem of deciding which agent wins in a reachability game can be effectively reduced to determining the existence of an $\emptyset$-NE in an explicit system $\mathbb{G}$. In this reduction, we are dealing with a succinctly represented reachability game and a circuit-based system. The construction of $\mathbb{G}$ in Lemma~\ref{explicitrealizabilitylowerboundlemma} and $\mathcal{G}$ in this proof are deliberately analogous, containing common elements such as ``penalty states" that disincentivize agents from making ``illegal" moves. Therefore, we can reapply the logic of Lemma~\ref{explicitrealizabilitylowerboundlemma}. In the reachability game version of $M$, the existential Agent $0$ attempts to reach a state corresponding to an accepting ID by choosing beneficial transitions when the head is in an existential state, and the universal Agent $1$ attempts to prevent this by choosing antagonistic transitions when the head is in a universal state. By applying  Lemma~\ref{explicitrealizabilitylowerboundlemma} to this succinctly represented reachability game, we can conclude that an $\emptyset$-NE exists in $\mathcal{G}$ iff Agent $1$ has a winning strategy in the reachability game representation of $M$. Suppose an $\emptyset$-NE strategy profile $\pi = \langle \pi_0, \pi_1 \rangle$ exists. By Lemma~\ref{explicitrealizabilitylowerboundlemma}, this means that $\pi_1$ is a winning strategy in the reachability game representation of $M$. Therefore, no matter how the transitions of $M$ are resolved at the existential states, it must be the case that when the universal states are resolved according to $\pi_1$, an accepting ID is never seen in $M$. This means that regardless of the transitions chosen at the existential states, the computation tree of $M$ must have at least one non-accepting branch (the one corresponding to the universal transitions being resolved according to $\pi_1$), and so $M$ must reject.
    
Now suppose that no such $\emptyset$-NE exists. This means that there is no winning strategy $\pi_1$ for Agent $1$ in the reachability game representation of $M$. Since there is no Agent $1$ winning strategy, it must be the case that there is an Agent $0$ winning strategy $\pi_0$. The logic now flows the same way. A winning strategy for Agent $0$ corresponds to transitions being made at the existential states, such that, regardless of the choices made at the universal states, an accepting ID must be eventually reached. Therefore, by resolving the existential transitions according to $\pi_0$, it must be the case that all branches of the computation tree of $M$ are accepting, and so $M$ accepts.
    
\end{proof}

The correctness of the reduction gives us our hardness result.

\begin{manualtheorem}{4.6}
    For the circuit-based model, the $W$-NE realizability problem is EXPTIME-complete. 
\end{manualtheorem}

\subsection{Construction of $\pi$ and Explicit Verification Upper Bound}\label{explicitverificationupperboundappendix}

This section corresponds to Section~\ref{upperboundexplicitverification}.

Given $k$ individual strategies represented by $k$ individual transducers, $\pi_i = \langle S_i , s^i_0 , D, \gamma^i, O_i \rangle$, we construct the cross-product transducer $\pi = \langle S, s_0, D, \gamma, O \rangle$ as follows:
\begin{enumerate}
    \item $S$, the set of states in $\pi$, is the cross-product of the set of states from each $\pi_i$. Therefore, $S  = \bigtimes_{i \in \Omega} S_i$. $s_0 \in S$, the initial state, corresponds to the tuple consisting of each initial state, i.e. $s_0 = \bigtimes_{i \in \Omega} \{s^i_0\}$.
    \item The input alphabet $D$ remains unchanged.
    \item The transition function $\gamma : S \times D \rightarrow S$ works component-wise. For a state $\langle s^0, s^1 \ldots s^{k-1} \rangle  \in S$ (where each $s^i \in S_i$), upon reading a character $d \in D$, the transducer transitions to $\langle \gamma^0(s^0, d),\gamma^1(s^1, d)  \ldots \gamma^{k-1}(s^{k-1},d )\rangle$.
    \item The output function is similarly defined component-wise. $O(\langle s^0, s^1 \ldots s^{k-1} \rangle) =  \\ \langle O_0(s^0), O_1(s^1) \ldots O_{k-1}(s^{k-1})\rangle$. An important point to note is that since $O_i(s^i) \in A_i$ for every agent $i$, we have $O(s) \in \bigtimes_{i \in \Omega} A_i = D$. 
\end{enumerate}

The cross-product representation of $\pi$ has exponential size with respect to the input strategies, as it involves taking an unbounded cross-product of the $k$ individual $\pi_i$ transducers, which were given as input. We now take the cross-product of $\pi$ with $\mathbb{G}$ to create $\mathbb{G} \times \pi = \langle V \times S, \langle v_0,s_0 \rangle, \tau, \gamma, G = \{ G_0 \ldots G_{k-1} \} \rangle$. This simple transition system analyzes the outcome of all agents following $\pi$ in $\mathbb{G}$. The set of states is the cross product $V \times S$, with the initial state $\langle v_0, s_0\rangle$, the pair of initial states of both $\mathbb{G}$ and $\pi$. At a given state $\langle v,s\rangle \in V \times S$, the system transitions to $\langle \tau(v, O(s)) , \gamma(s,O(s)) \rangle$, representing the outcome of all agents playing the actions suggested by $\pi$ and each of the two components transitioning accordingly. The goal sets $G_i$ are carried over from the original specification of $\mathbb{G}$ in order to easily recreate the reachability goal of Agent $i$ in $\mathbb{G} \times \pi$ as the set $\{\langle v,s \rangle | v \in G_i \}$. Note that since $\mathbb{G} \times \pi$ can be viewed a Markov Chain (a probabilistic transition system, see~\cite{Puterman94} for more details) in which all transitions are deterministic, the execution of $\mathbb{G} \times \pi$ corresponds to a unique trace -- the primary trace of $\pi$ in $\mathbb{G}$.

Since the $W$-Nash equilibrium solution concept involves testing whether an agent has a profitable unilateral deviation or not, it is also important to analyze what happens when every agent but one follows $\pi$. Essentially, this would be $\mathbb{G} \times \pi$ with one ``free agent" that selects their own actions instead of the ones suggested by $\pi$. To this end we construct
a deterministic Markov Decision Process (MDP) (a Markov Chain augmented by a single agent that chooses actions that influences the transitions, see~\cite{Puterman94} for more details) $\mathbb{G}_i \times \pi =  \langle V \times S, \langle v_0,s_0 \rangle, \tau,\gamma, A_i, G = \{ G_0 \ldots G_{k-1} \} \rangle$, which augments $\mathbb{G} \times \pi$ with the set of actions $A_i$,  to analyze the possibilities that occur when Agent $i$ deviates. This system functions largely the same as $\mathbb{G} \times \pi$, the only difference is that Agent $i$ is allowed to choose his action at each state as opposed to the one suggested by $\pi_i \in \pi$. Basically, in $\mathbb{G} \times \pi$, at a state $\langle v,s \rangle \in V \times S$, transition always occurs as if the agents had selected $\langle O_0(s^0) \ldots O_i(s^i) \ldots  O_{k-1}(s^{k-1}) \rangle$ as their collective decision. In $\mathbb{G}_i \times \pi$, Agent $i$ is allowed to choose the action in the $i$-th slot regardless of what is suggested by $O_i(s^i)$. Therefore, if they choose action ``$a$" instead, then transition proceeds as if $\langle O_0(s^0) \ldots a \ldots  O_{k-1}(s^{k-1}) \rangle$ were chosen in $\mathbb{G} \times \pi$ instead of $\langle O_0(s^0) \ldots O_i(s^i) \ldots  O^{k-1}(s^{k-1}) \rangle$. The effect of such a change is trivially computable given $\tau$ and $\gamma$.

With the construction of $\mathbb{G} \times \pi$ and $\mathbb{G}_i \times \pi$ in hand, we can move on to considering the verification problem. Essentially, it consists of two different types of queries (see also~\cite{RV22}): 
\begin{enumerate}
    \item Is it the case that when $\pi$ is followed in $\mathbb{G}$, precisely the agents in $W$ meet their goal?
    \item Is it the case the for agents $j \not \in W$, they have no unilateral profitable deviation from $\pi$?
\end{enumerate}
These queries directly test the conditions of the $W$-NE solution concept.

 Before discussing the complexity of the two queries, it is important to recall that $\mathbb{G} \times \pi$ has exponential size w.r.t. the input strategies due to the presence of the unbounded cross-product construction $\pi$. Therefore, we never explicitly construct $\mathbb{G} \times \pi$ (or $\mathbb{G}_i \times \pi$); instead, we reason about these constructions through on-the-fly algorithms. Furthermore, it is also important to note that since $\pi$ is a finite-state transducer, the primary trace of $\pi$ in $\mathbb{G}$ must be eventually periodic, i.e., ending in an infinitely-repeating cycle.

For the first type of queries, we follow the primary trace of $\pi$ by following the transitions in $\mathbb{G} \times \pi$ state-by-state. Once we compute the successor of a state $\langle v,s \rangle \in \mathbb{G} \times \pi$, we delete the previous state from memory and store the successor state. In this way, we can follow the primary trace of $\pi$ by reusing space to store at most two states at once. This ``following" algorithm can be performed in PSPACE, as states in an exponentially large set $V \times S$ have polynomial-size encodings. At some point, the algorithm uses non-determinism to ``guess" that it has entered the periodic portion of the primary trace. It does this by guessing the state at the start of the cycle and then verifying that $\mathbb{G} \times \pi$ eventually returns to this state. Once $\mathbb{G} \times \pi$ returns to a state, its deterministic transitions ensure that no additional states will be visited on the infinite primary trace. This ``cycle-guessing" portion of the algorithm can be performed in NPSPACE=PSPACE~\cite{VW94}, as it uses non-determinism but still stores a constant number of states. It is then trivial to check if an Agent's goal is eventually visited on the primary trace of $\pi$, giving us a PSPACE upper bound for the first type of queries.

The second type of queries amounts to seeing if an Agent $j \not \in W$ has a deterministic strategy to reach a state $\langle v,s\rangle$ with $v \in G_j$ in $\mathbb{G}_j \times \pi$. Since $\mathbb{G}_j \times \pi$ is an MDP, one of the optimal strategies in $\mathbb{G}_j \times \pi$ must be a \emph{policy}, a strategy that is both deterministic and memoryless~\cite{Puterman94}. Note that since all of the transitions in $\mathbb{G}_i \times \pi$ are deterministic, if it is possible to reach a set with positive probability in $\mathbb{G}_j \times \pi$, then it is possible to reach it with probability $1$. Therefore, we can use non-determinism to guess the actions recommended by this optimal policy state-by-state. So, at each state, we non-deterministically guess an action from $A_i$ and then transition according to $\tau$ and $\gamma$ in $\mathbb{G}_j \times \pi$. As argued before, this algorithm is in NPSPACE = PSPACE. Furthermore, since PSPACE is closed under complementation~\cite{papacomplexity}, we can also test whether Agent $j$ \emph{cannot} reach a state $\langle v,s\rangle$ with $v \in G_j$ in PSPACE as well. This gives us a PSPACE upper bound for both families of queries, and therefore, a PSPACE upper bound for the verification problem.

\begin{manualtheorem}{5.1}
    For the explicit model, the $W$-NE verification problem is in PSPACE.
\end{manualtheorem}

\subsection{Explicit Verification Lower Bound Reduction}\label{explicitlowerboundverificationappendix}

This section corresponds to Section~\ref{lowerboundexplicitverification}.

As a reminder, the problem we are reducing from is:
\begin{problem*}
    Given a deterministic Turing Machine $M$ and a natural number $n \in \mathbb{N}$ written in unary, does $M$ accept the empty tape using at most $n$ cells?
\end{problem*}

This problem is the deterministic counterpart to the one considered in Section~\ref{circuitbasedrealizabilitylowerboundappendix} and is a canonical PSPACE-complete problem~\cite{sipser2006}.  We now show that this problem reduces to the problem of $W$-NE verification for explicit systems, using the notation $M = \langle R , \Gamma, r_0 ,\Delta, F \rangle $ for $R$ a set of states, $\Gamma$ an alphabet, $r_0 \in R$ an initial state, $\Delta$ a transition function, and $F \subseteq R$ a set of final states. $M$ is specified in almost the same way as the alternating Turing Machine $M$ in Definition~\ref{altmachine}, with only two key differences. The first is that the transition function $\Delta$ is now deterministic, so it simply has type $\Delta : R \times \Gamma \rightarrow R \times \Gamma \times \{\rightarrow,\leftarrow \}$. The second is that there is no need to consider universal, existential, or rejecting states, so we simply have a set $F \subseteq R$ of accepting states. 

We now present the reduction, which reduces the space-bounded Turing Machine acceptance problem to the verification of an $\Omega$-NE in a system $\mathbb{G} = \langle V, v_0, \Omega = \{0 \ldots k-1\}, A = \{A_0 \ldots A_{k-1} \}, G = \{ G_0 \ldots G_{k-1} \}, \tau \rangle$, where the input strategies are specified by finite-state transducers of the form $ \pi_i =\langle S_i , s^i_0 , D, \gamma^i, O_i \rangle$. As before, the precise definitions of the components of $\mathbb{G}$ and $\pi = \langle \pi_0 \ldots \pi_{k-1} \rangle$ are given ``on-the-fly" in the proof to increase readability.

\begin{manuallemma}{5.2}
    The machine $M$ accepts the empty tape in at most $n$ iff the strategy profile $\pi$ is an $\Omega$-NE in $\mathbb{G}$.
\end{manuallemma}

\begin{proof}

The main idea behind this reduction is to simulate the IDs of the Machine $M$ in a distributed manner by using the state spaces $S_i$ of the transducers in the tuple $\pi = \langle \pi_0 \ldots \pi_{k-1} \rangle$. 

At a lower level, Agent $i$ uses the state-space $S_i$ of the strategy $\pi_i$ to keep track of the $i$-th character of the ID (see Definition~\ref{machineIDdef} and Example~\ref{IDexample} about machine IDs). In order to do so, they need information about the $i-1$st, $i$th, and $i+1$st characters of the previous ID (with trivial exceptions for Agent $0$, who does not have a $i-1$st cell, and Agent $k-1$, who does not have an $i+1$st cell). In this way, Agent $i$ keeps track of a ``subID" of the machine -- instead of accounting for the entire contents of the tape, they track the characters in the $i-1$st, $i$th, and $i+1$st position. The state space of $\pi_i$ is then the set of all possible subIDs. It is crucial to note that each subID consists of a fixed number of characters, as this means the total number of subIDs is polynomial in the size of the input. The initial state of each $\pi_i$ then corresponds to the relevant subID of the initial configuration of the tape, which is empty except for the head, which we assume is positioned over the leftmost cell. 

Let us now consider an example where $\pi_i$ is in state $0,\langle 1,r \rangle,0$. Using our notation for machine IDs (Example~\ref{IDexample}) to subIDs, this state means that the $i-1$st cell reads $0$, the $i$th cell reads $1$, and the $i+1$st cell reads $0$. Furthermore, the head is positioned over cell $i$ and is in state $r$. The output function $O_i$ then provides the details of the transition of $M$ as an element of $R \times \Gamma \times \{\rightarrow, \leftarrow\}$, which specifies which state $M$ transitions to, which character is written in cell $i$, and which direction the head moves in. This set $R \times \Gamma \times \{\rightarrow, \leftarrow\}$ the set of actions $A_i$ for Agent $i$. Since this game is turn-based, the decision set $D$ is no longer exponential in the number of agents; since only one agent chooses an action at a time, there is no need to consider a joint action space. Therefore, the transition table for each $\pi_i$ can be represented in polynomial space. 

This brings us to the multi-agent system $\mathbb{G}$, which, as mentioned before, is a turn-based game. It has a state space $V = V_1 \cup V_2 \cup V_3$ consisting of three components. The first component, $V_1$, consists of $n$ states labeled by $0 \ldots n-1$. The next component $V_2$ is the cross product of the set of agents $[n]$ and a set of actions $A_i$, given by  $[n] \times R \times \Gamma \times \{\rightarrow, \leftarrow\}$. The final component $V_3$ consists of two special states labeled ``accept" and ``reject".

The dynamics of $\mathbb{G}$ are relatively straightforward. When the head is over cell $i$ in $M$, then the multi-agent system $\mathbb{G}$ is in state $i \in V_1$. The game is turn-based, so state $i$ is owned by Agent $i$. Agent $i$ uses the transducer $\pi_i$ to output an action corresponding to the relevant transition information of the machine $M$. 

The multi-agent system $\mathbb{G}$ then transitions to the state corresponding to the cross product of $i$ and the action that was just output. This allows the other agents to observe the transition that was made at cell $i$ and update their own $\pi_i$ state if necessary. For example, suppose that Agent $i$ used his action to announce that the head was moving right in state $r$ and that the character $0$ was written at cell $i$. This means that cell $i$ and cell $i+1$ will update their values in the ID, providing relevant information to update the states of $\pi_{i-1},\pi_i,\pi_{i+1},$ and $\pi_{i+2}$. If, however, the change occurred at $i = 5$, then Agent $0$, for example, does not need to update the state of $\pi_0$, as $S_0$ only tracks the contents of cell $0$ and $1$. The system then transitions from the state in $V_2$, reflecting the transition information, to the state $i + 1 \in V_1$, reflecting the new position of the head (which has moved rightwards from cell $i$). Transitions from states in $V_2$ ignore the action choices of the agents. In our example, it is now Agent $i+1$'s turn to select the action output by $\pi_{i+1}$. We assume that the head starts at the leftmost cell, and so the initial state is the cell $0$ in $V_1$.

Transitions continue to occur in this manner with two exceptions. First, if the head is ever in a state in $F$, then the computation has accepted, and the system transitions to the state accept $\in V_3$. This state ``accept" is the reachability goal of every agent. Second, if the head ever moves out of bounds, for example, if it is in a state $v \in V_2$ that suggests moving the head left from position $0$, then the system transitions to the sink ``reject" state, from which it is no longer possible to reach the reachability goal state ``accept". 

It is then straightforward enough to see that $M$ accepts the empty tape using at most $n$ cells iff $\pi$ is an $\Omega$-NE in $\mathbb{G}$. Together $\mathbb{G}$ and $\pi$ simulate the computation of $M$. Since we are asking for the existence of an $\Omega$-NE, all agents have to meet their goal on the primary trace of $\pi$ in $\mathbb{G}$ (since they all have the same reachability goal). Note that the verification of an $\Omega$-NE is reasoned about entirely through the primary trace, as there are no agents $j \not \in W$ that consider deviation. Furthermore, note that all constructions in $\mathbb{G}$ and $\pi$ can be done in polynomial time. The state space of a strategy $\pi_i$ is at most $|\Gamma \times R|^3$, the set $V$ of states consists of polynomial-sized constructions from $M$, etc. The most critical point here is $\tau$, the transition table. Just as we did for the transducers $\pi_i$, we take explicit advantage of the fact that this is a turn-based game (see Definition~\ref{reachabilitygamedef}), employing the same strategy as~\cite{BBMU15} and other works. This means that the number of rows in $\tau$ has a polynomial upper bound of $|V| \cdot |[n] \times R \times \Gamma \times \{ \rightarrow , \leftarrow\}|$, making this a valid polynomial-time reduction.  
\end{proof}

\begin{manualtheorem}{5.3}
        For the {\bf turn-based} explicit model, the $W$-NE verification problem is PSPACE-complete.
\end{manualtheorem}

\subsection{Circuit-Based Verification Lower Bound Reduction}\label{circuitlowerboundverificationappendix}

This section corresponds to Section~\ref{lowerboundcircuitverification}.

We reduce from the same problem as in Section~\ref{explicitlowerboundverificationappendix}.

\begin{problem*}
    Given a deterministic Turing Machine $M$ and a natural number $n \in \mathbb{N}$ written in unary, does $M$ accept the empty tape using at most $n$ cells?
\end{problem*}

As in Section~\ref{circuitbasedrealizabilitylowerboundappendix}, our reduction does not focus on low-level details of circuit construction, and the constructions of the system and strategy are done on-the-fly to increase readability. We once again use the same specification for $M = \langle R , \Gamma, r_0 ,\Delta, F \rangle $ as in Section~\ref{explicitlowerboundverificationappendix}. In the reduction here, however, we use a one-agent system $\mathcal{G} = \langle {\tt V}, {\tt v_0 }, \Omega = \{0\},   {\tt A} = \{ {\tt A_0} \}, {\tt G } = \{ {\tt G_0} \} , \varphi \rangle$ and a strategy profile consisting of a single strategy $\pi = \langle \pi_0 \rangle$. Since there is only one strategy $\pi_0$, we write $\pi=\pi_0 = \langle \Salt , \s_0 , {\tt D}, \omega, \Oalt \rangle$, removing the $i$-indexing (see Definition~\ref{circuittransducermodel}). As before, we give the precise definitions of these components on-the-fly in order to increase readability. 

\begin{manuallemma}{5.5}
    The machine $M$ accepts the empty tape in at most $n$ iff the strategy profile $\pi$ is an $\Omega$-NE in the one-agent system $\mathcal{G}$.
\end{manuallemma}
\begin{proof}
     Just as in Section~\ref{explicitlowerboundverificationappendix}, we use the IDs of the machine $M$ to track the computation of $M$. As further discussed in Section~\ref{circuitbasedrealizabilitylowerboundappendix}, the IDs of a machine $M$ (see Definition~\ref{machineIDdef} and Example~\ref{IDexample} about machine IDs) that can use at most $n$ cells is given by the set of $n$-length strings over the alphabet $(R \cup \{\bot\}) \times (\Gamma \cup \{\bot\})$.

    The state set in $\mathcal{G}$ is expressed using a set $V$ of boolean variables, which encode the set of $n$-length strings over $(R \cup \{\bot\}) \times (\Gamma \cup \{\bot\})$. There are $(|(R \cup \{\bot\})| \times |\Gamma \cup \{\bot\}|)^n$ such strings, meaning that this state set can be encoded using polynomially many boolean variables. Furthermore, unlike Section~\ref{explicitlowerboundverificationappendix}, here $\mathcal{G}$ only has a single agent.

    The output function ${\tt O}$ of the single strategy $\pi_0$ is  a circuit encoding of $\Delta$, with 
    $\Salt = { \tt V}$. 
    Intuitively, the state space $2^{\Salt}$ of $\pi_0$ is also the set of IDs. When in a state corresponding to a specific ID, ${\tt O}$ locates the position of the head and outputs the transition information of $\Delta$.  This circuit represents a discrete function that can be represented by a table that has $R \times (\Gamma \cup \{\bot\})$ rows, meaning it can be constructed in polynomial time. The set of actions $\tt{A}_0$ is a boolean encoding of the set $[n] \times R \times \Gamma \times \{\rightarrow, \leftarrow\}$, the range of $\Delta$ augmented by a position on the tape given by $[n] = \{0 \ldots n-1\}$, as we specify that ${\tt O}$ also identifies the current position of the head in parallel while computing the output of $\Delta$. This output action then informs the transition function $\varphi \in \mathcal{G}$ on how to update the state ${\tt v} \in 2^{ \tt V}$ (which represents the current ID of $M$) by providing $\varphi$ with the information of where the transition occurred (through $[n]$) and how the tape changed (through the movement and state of the head given by $R \times \{\rightarrow, \leftarrow\}$ and the character written at cell $i \in [n]$ given by $\Gamma$). 
    
    The transition function $\varphi \in \mathcal{G}$ computes the new state corresponding to the new ID of the machine $M$. This involves changing the contents of at most two cells at a time, so it can be computed by a polynomial-size circuit, giving $\varphi$ a polynomial-sized representation. If the computation makes an error, such as moving left when the head is at cell $0$, then the system transitions to a new sink state, just as before. As in Section~\ref{explicitlowerboundverificationappendix}, inputs of actions that do not correspond to the output of $\Delta$ are not a concern, as the verification problem is specifically asking about the behavior of $\pi_0$, a strategy that faithfully recreates the behavior of $\Delta$. Similarly, the transition function $\omega$ of $\pi_0$ mirrors the transition function of $M$ and $\mathcal{G}$. When the system $\mathcal{G}$ moves to a state ${ \tt v } \in { \tt V}$, the transition function $\omega : 2^{\Salt} \times 2^{\tt D} \rightarrow 2^{\Salt} $ moves to the same state in $2^{\Salt}$, as ${\tt V} = \Salt$. Both initial states ${\tt v_0}$ and ${\tt s_0}$ correspond to the starting ID of $M$. 

    The reachability goal of the only agent, Agent $0$, is the set of states ${\tt v} \in 2^{ \tt V}$ that represent an ID in which the head of $M$ is in a final state $r \in F$.  Such a state can be recognized by building a circuit that computes the characteristic function of $F$. 

    As before, the question is then whether $\pi$ represents an $\Omega$-NE, or, in this case, a $\{0\}$-NE. The interaction of $\pi$ and $\mathcal{G}$ computes the IDs of $M$ in a very straightforward way, meaning that if $\pi$ is an $\Omega$-NE then Agent $0$ reaches a state that corresponds to an accepting ID and $M$ accepts the empty tape using at most $n$ cells. This gives us our PSPACE-hard reduction.
\end{proof}

\begin{manualtheorem}{5.6}
For the circuit-based model, the $W$-NE verification problem is PSPACE-complete. 
\end{manualtheorem}

\section{Conclusion and Future Work}\label{Conclusion}
The main point of this paper is that the circuit-based model, as opposed to the explicit model used in the literature, is the right model to use for more precise complexity-theoretic results. The main arguments in support of this claim can be found in the discussion sections, Section~\ref{realizabilitydiscussion} for the realizability problem and Section~\ref{verificationdiscussion} for the verification problem. These discussion sections outline the mathematical issues that arise when the explicit model is used and how the circuit-based model addresses them. 

While this paper is most interested in the mathematical justification behind using the circuit-based model over the explicit one, there are other compelling reasons as well. The field of equilibrium analysis was broadly motivated by a need to mathematically analyze software systems that consisted of multiple independent parts. The use of the explicit table moved away from this motivation by introducing a design pattern that would never appear in production software -- a monolithic table structure that explicitly ranges over a very large set of tuples. Therefore, we see the circuit-based model as a unique opportunity to advance the mathematical rigor and precision of a field alongside its practical applications simultaneously. It is this broad idea of creating more precise and applicable multi-agent systems that motivates our directions of future work. Here, we outline two important directions to consider.

The first direction for future work involves developing probabilistic extensions of the deterministic models considered in this paper. Probability is playing an increasingly important role in software systems~\cite{aibook}, and the mathematical models used in equilibrium analysis should reflect that. Therefore, developing succinct mathematical representations for probabilistic transition functions and strategies is an important direction for future work.  

In this paper, we studied how different parts of the input in a multi-agent system influence the complexity of the verification and realizability decision problems, broadly classifying certain components as ``driving" the total complexity of the results. Furthermore, we factored out the representation of agents' goals from the overall complexity-theoretic results by considering only reachability goals. An important next step for future work is to develop algorithms specifically suited to certain classes of multi-agent systems. For example, different algorithms could be developed to be specifically performant on systems with a relatively large number of agents compared to the size of the transition table, or vice versa. There is also a salient question of designing algorithms that are performant when reasoning about the representation of a single agent's goal dominates the overall complexity of equilibrium decision problems. This case is especially important to the equilibrium-analysis literature, which often considers settings with temporal-logic goals (such as LTL), which end up dominating the overall complexity of the equilibrium decision problem. For an explicit discussion on this point, see~\cite{RV21,RV22}.

\bibliographystyle{alphaurl}
\bibliography{bib}

\end{document}